\documentclass[11pt]{article}
\usepackage{color,amsmath, amsthm, amsfonts, amssymb, dsfont, graphicx, listings, graphics, epsfig, enumerate, bbm,dsfont,setspace, systeme}

\usepackage[margin=1in]{geometry}
\usepackage[fleqn,tbtags]{mathtools}
\usepackage[english]{babel}
\usepackage[colorinlistoftodos]{todonotes}
\usepackage[normalem]{ulem}

\usepackage[numbers]{natbib}   

\title{Decomposable sums and their implications on\\ naturally quasiconvex risk measures}
\author{\c{C}a\u{g}{\i}n Ararat\thanks{Bilkent University, Department of Industrial Engineering, Ankara, Turkey, cararat@bilkent.edu.tr.}
\and
Bar{\i}\c{s} Bilir\thanks{University of Texas at Austin, Department of Operations Research and Industrial Engineering, Austin, TX, barisbilir@utexas.edu.}
\and
Elisa Mastrogiacomo\thanks{Universit\`a degli Studi dell'Insubria, Department of Economics, Varese, Italy, elisa.mastrogiacomo@uninsubria.it.}}
\date{\today}

\addto\captionsenglish {}
\makeatletter \renewenvironment{proof}[1][\proofname] {\par\pushQED{\qed}\normalfont\topsep6\p@\@plus6\p@\relax\trivlist\item[\hskip\labelsep\bfseries#1\@addpunct{.}]\ignorespaces}{\popQED\endtrivlist\@endpefalse} \makeatother
\newtheorem{thm}{Theorem}[section]

\newtheorem{lem}[thm]{Lemma}
\newtheorem{prop}[thm]{Proposition}
\newtheorem{assumption}[thm]{Assumption}
\newtheorem{defn}[thm]{Definition}
\theoremstyle{definition}
\newtheorem{example}[thm]{Example}
\newtheorem{rem}[thm]{Remark}

\numberwithin{equation}{section}
\newcommand{\R}{\mathbb{R}}

\newcommand{\cF}{\mathcal{F}}
\newcommand{\cG}{\mathcal{G}}
\newcommand{\N}{\mathbb{N}}
\newcommand{\bN}{\mathbb{N}}
\newcommand{\bbM}{\mathbf{M}}
\newcommand{\bbN}{\mathbf{N}}
\newcommand{\sm}{\!\setminus\!}

\DeclareMathOperator{\cl}{cl}

\DeclareMathOperator{\spn}{span}
\DeclareMathOperator{\dom}{dom}

\DeclareMathOperator{\esssup}{ess\,sup}

\DeclareMathOperator{\Leb}{Leb}

\newcommand{\G}{\mathcal{G}}
\newcommand{\F}{\mathcal{F}}
\newcommand{\X}{\mathcal{X}}
\newcommand{\B}{\mathcal{B}}
\newcommand{\Y}{\mathcal{Y}}

\renewcommand{\P}{\mathbb{P}}
\newcommand{\E}{\mathbb{E}}

\renewcommand{\H}{\mathcal{H}}

\newcommand{\of}[1]{\ensuremath{\left( #1 \right)}}
\newcommand{\cb}[1]{\ensuremath{ \left\{ #1 \right\} }}
\newcommand{\sqb}[1]{\ensuremath{ \left[ #1 \right] }}
\newcommand{\norm}[1]{\ensuremath{ \left\Vert #1 \right\Vert }}
\newcommand{\ip}[1]{\ensuremath{ \left\langle #1 \right\rangle }}

\def\prehp(#1,#2){\ensuremath{  #1 \cdot #2 }}

\usepackage[latin1]{inputenc}
\usepackage{tikz}
\usetikzlibrary{trees}
\usepackage{pgfplots,caption}
\usepackage{tikz-qtree}
\usepackage{xcolor}
\usetikzlibrary{decorations.pathreplacing}

\usepackage{pgfplots}
\pgfplotsset{compat=1.9}
\usepackage{tikz}
\usetikzlibrary{trees}
\usepackage{pgfplots,caption}
\usepackage{tikz-qtree}
\usepackage{xcolor}
\usetikzlibrary{decorations.pathreplacing}
\usetikzlibrary{arrows, decorations.markings}
\tikzstyle{vecArrow} = [thick, decoration={markings,mark=at position
	1 with {\arrow[semithick]{open triangle 60}}},
double distance=1.4pt, shorten >= 5.5pt,
preaction = {decorate},
postaction = {draw,line width=1.4pt, white,shorten >= 4.5pt}]

\begin{document}
\maketitle
\thispagestyle{empty}
	
\begin{abstract}
 	Convexity and quasiconvexity are two properties that capture the concept of diversification for risk measures. Between the two, there is natural quasiconvexity, an old but not so well-known property weaker than convexity but stronger than quasiconvexity. A detailed discussion on natural quasiconvexity is still missing and this paper aims to fill this gap in the setting of conditional risk measures. We relate natural quasiconvexity to additively decomposable sums. The notion of convexity index, defined in 1980s for finite-dimensional vector spaces, plays a crucial role in the discussion of decomposable sums. We propose a general treatment of convexity index in topological vector spaces and use it to study naturally quasiconvex risk measures. We prove that natural quasiconvexity and convexity are equivalent for conditional risk measures on $L^{p}$ spaces, $p \geq 1$, under mild continuity and locality conditions. Finally, we discuss an alternative notion of locality with respect to an orthonormal basis in $L^2$.
 	\\
 	\\[-10pt]
 	\textbf{Keywords and phrases: }convexity index, decomposable sum, natural quasiconvexity, risk measure\\
 	\\[-10pt]
 	\textbf{Mathematics Subject Classification (2020): }46N10, 46B15, 52A01, 91G70
\end{abstract}


\section{Introduction}

Measuring the risk of a financial position is a central problem in finance. Risk measures are functionals that are used for this purpose. A static risk measure maps a random variable to its minimum deterministic capital requirement evaluated today. In a dynamic setting, the risk of a random variable can be evaluated at an intermediate time, in which case the capital requirement is also random. Such functionals are called conditional risk measures. In the literature, a lot of emphasis has been given to two diversity-related properties of static and conditional risk measures: convexity and quasiconvexity. The aim of this paper is to investigate the so-called natural quasiconvexity for conditional risk measures, which is a property in abstract convexity that is not very well-understood.

In the seminal paper \citet{Artzner}, coherent risk measures are defined as real-valued functionals on a vector space of random variables that model financial positions. Under the axioms of a coherent risk measure, having greater returns for all possible scenarios implies lower risk (monotonicity), a deterministic amount added to the position reduces the risk by the same amount (translativity or cash additivity), and risk increases in a sublinear way (subadditivity and positive homogeneity).

Coherent risk measures are then generalized to convex risk measures in \citet{Follmer}, \citet{Frittelli-1}. The motivation behind this generalization is that the risk of a financial position may increase in a nonlinear way with the size of a position. Therefore, subadditivity and positive homogeneity axioms are replaced with the weaker convexity axiom. It is noteworthy that convexity explicitly captures the idea that diversification does not increase risk. Later, it is argued in \citet{El-Karoui} that cash additivity should be replaced with cash sub-additivity since the former ignores the uncertainty on interest rates. Moreover, it is a well-known result that quasiconvexity and convexity are equivalent under cash additivity (see, e.g., \citet[Corollary 4.2]{Marinacci}). Hence, the replacement of cash additivity with cash sub-additivity paves the way for drawing the distinction between the quasiconvexity and convexity properties of risk measures. Important works along these lines are \citet{Cerreia-1}, \citet{Drapeau}, \citet{Frittelli-2}. Furthermore, it is argued in \citet{Cerreia-2} that quasiconvexity is indeed the right mathematical formulation of diversification under cash sub-additivity. For applications of quasiconvex risk measures, we refer the reader to \citet{Elisa}, \citet{Kallblad}, \citet{Mucahit}.

In the conditional setting, a risk measure gives the capital requirement for a financial position at an intermediate time. In this case, the functional maps into a vector space of random variables that are measurable with respect to a smaller $\sigma$-algebra. For discussions on conditional risk measures, we refer the reader to \citet{Bion-Nadal}, \citet{Detlefsen}, \citet{Frittelli-2}, \citet{Riedel}, \citet{Frittelli-3}, \citet{Ruszczynski}. 

Let $(\Omega, \F, \P)$ be a probability space, $\G \subseteq \F$ a sub-$\sigma$-algebra, and $p\geq 1$. A conditional risk measure $\rho \colon L^{p}(\Omega, \F, \P) \to L^{p}(\Omega, \G, \P)$ is called naturally quasiconvex if, for every $X,Y\in L^{p}(\Omega, \F, \P)$ and $\lambda \in [0,1]$, there exists $\mu\in [0,1]$ such that
\[
	\rho(\lambda X + (1-\lambda)Y) \leq \mu \rho(X) + (1-\mu) \rho(Y),
\]
where the inequality is understood in the almost sure sense. Natural quasiconvexity is defined in more abstract settings for vector-valued and set-valued functions in 1990s; see \citet{Helbig} (under the name fractional convexity), \citet{Tanaka}, \citet{Kuroiwa}. Clearly, natural quasiconvexity is stronger than quasiconvexity but weaker than convexity. Furthermore, natural quasiconvexity is equivalent to a property called $\star$-quasiconvexity, which is defined as the quasiconvexity of a family of scalar functions induced by dual elements of the image space. In view of this equivalence, it turns out that naturally quasiconvex risk measures are closely related to additively decomposable sums, which we discuss next.

In \citet{Debreu}, and \citet{Crouzeix}, an in-depth discussion on additively decomposable sums can be found. In their setting, a real-valued function $s$ defined on the product $\X_{1} \times \ldots \times \X_{n}$ of $n$ open convex factor sets $\X_1,\ldots,\X_n$ of arbitrary finite dimension is called additively decomposable if 
\[
	s(x_{1}, \ldots, x_{n}) = f_{1}(x_{1}) + \ldots +f_{n}(x_{n}),\quad x_1\in \X_1,\ldots,x_n\in\X_n,
\]
for some coordinate functions $f_1,\ldots,f_n$. They define the convexity index of each coordinate function as a break-even point $\lambda\in\R$ at which a suitable exponential tranformation of the function indexed by $\lambda$ becomes convex/concave. They study the properties of this index and show that the function is convex if and only if its convexity index is nonnegative. Then, they characterize the quasiconvexity of an additively decomposable function $s$ in terms of the nonnegativity of the sum of the
convexity indices of its coordinate functions $f_1,\ldots,f_n$. As a by-product, they prove that if $f_1,\ldots,f_n$ are not constant, then all of them, with at most one
exception, are actually convex if and only if $s$ is quasiconvex. 

Potential applications of these results can be found in some areas of economic theory. For example, techniques developed in the study of quasiconcave
additively decomposable functions (which are almost concave) appear in utility theory, see, e.g., \citet{De60}, \citet{Gr61}, \citet{Ra72}, and the more recent papers \citet{Wak94}, \citet{GhiMa01}. 
In addition, additive decomposability is used in the context of production functions, see, among the others, \citet{ArEn61}.

In this paper, we propose a novel application of decomposable sums in the field of naturally quasiconvex risk measures. Since risk measures are defined on Lebesgue spaces of random variables, which are generally infinite-dimensional, the available results in \citet{Debreu}, and \citet{Crouzeix} are not applicable to our setting. Hence, we extend the framework of these papers fundamentally by considering general topological vector spaces. Later, this extension sheds light on natural quasiconvexity for conditional risk measures.

In the abstract setting of general topological spaces, we also consider infinite decomposable sums and characterize the quasiconvexity of such sums in terms of the convexity indices of the coordinate functions. To the best of our knowledge, the case of infinite sums has not been considered before even in the finite-dimensional setting.

The remainder of this paper is organized as follows. In Section \ref{sec:sums}, we study additively decomposable sums on topological vector spaces. We introduce the convexity index for extended real-valued functions defined on general vector spaces and present some important properties of it in Subsection \ref{subsec:index}. Then, in Subsection \ref{subsec:finsum}, we show that an additively decomposable finite sum is quasiconvex if and only if either all functions that appear in the sum are convex or all except one are convex together with a condition on the sum of the convexity indices of them. It is noteworthy that Subsections \ref{subsec:index} and \ref{subsec:finsum} have strong links with \citet{Debreu}, and \citet{Crouzeix} since most of the results in these sections are generalizations of their results to real-valued functions defined on general topological vector spaces. To make this generalization work, it appears that we only need a lower semicontinuity assumption on each function that appears in the additively decomposable sum. In Subsection \ref{subsec:infsum}, we show that almost the same result with the one in the previous section holds for infinite decomposable sums. In Section \ref{sec:nqc}, we study naturally quasiconvex conditional risk measures. Subsection \ref{subsec:risk} is a brief introduction to risk measures. In Subsection \ref{NQC defn}, we discuss natural quasiconvexity and give an equivalent characterization of it. In Subsection \ref{subsec:rel}, we work on a general $L^p$ space with $p\geq 1$, and show that convexity and natural quasiconvexity are exactly the same properties for conditional risk measures, under some mild conditions. Finally, in Subsection \ref{subsec:L2}, we work on an $L^2$ space with a special structure on the underlying probability space, define a new property, namely, locality with respect to an orthonormal basis, and show that, under that property, naturally quasiconvexity and convexity are equivalent with respect to the preorder defined by the cone generated by the elements of the basis.


\section{Decomposable sums in general vector spaces}\label{sec:sums}

\subsection{Convexity index}\label{subsec:index}
Convexity index for real-valued functions on $\R^{n}$ is first introduced in \citet{Debreu}, and studied further in \citet{Crouzeix}. The main concern of this section is to extend the definition of convexity index to extended real-valued functions on general topological vector spaces. This will be the building block for the study of additively decoposable sums on such spaces.

As we work with extended real-valued functions, the following conventions for the arithmetic on $\bar{\R}=[-\infty,+\infty]$ are used throughout the paper. We have $0\cdot(+\infty)=0\cdot(-\infty)=0$,  $\frac{1}{0} = +\infty$. We also set $e^{z}=+\infty$ if $z=+\infty$ and $e^z=0$ if $z=-\infty$.

Let $\X$ be a topological vector space and $f\colon \X\to \bar{\R}$ a function which we keep fixed unless stated otherwise. We define the \emph{effective domain} of $f$ as the set $\dom f\coloneqq \cb{x\in \X \colon f(x)<+\infty}$. We assume that $f$ is proper in the sense that $\dom f \neq \emptyset $ and $f(x)>-\infty $ for every $x\in \X$. For each $\lambda\in\R$, we associate to $f$ the function $r_\lambda\colon \X\to\R$ defined by
\begin{equation}\label{rlambda}
	r_{\lambda}(x) \coloneqq e^{-\lambda f(x)},\quad x\in \X.
\end{equation}
First, we present an auxiliary result related to the convexity properties of $r_\lambda$, $\lambda\in\R$, which is helpful to have a better grasp of the definition of convexity index.

\begin{lem} \label{lemma 1 section 2}
	The following results hold for $r_\lambda$, $\lambda\in\R$, associated to $f\colon \X\to\bar{\R}$.
	\begin{enumerate}[(i)]
		\item Let $\lambda < 0$. Then, $r_{\lambda}$ is convex if and only if $r_{\mu}$ is convex for each $\mu < \lambda$.
		\item Let $\lambda > 0$. Then, $r_{\lambda}$ is concave if and only if $r_{\mu}$ is concave for each $\mu\in [0,\lambda)$.
		\item If $r_{\lambda}$ is concave for some $\lambda > 0$, then $r_{\mu}$ is convex for each $\mu < 0$.
	\end{enumerate}
\end{lem}

\begin{proof}
	Let $\lambda \neq 0$ and $\mu < \lambda$. Then, we may write $r_{\mu} = k_{\mu,\lambda}\circ r_{\lambda}$, where $k_{\mu,\lambda}(t) \coloneqq  t^{\frac{\mu}{\lambda}}$ for each $t\in [0, +\infty]$ with the conventions $0^0=(+\infty)^0=1$, $(+\infty)^a=+\infty$ for $a>0$, and $(+\infty)^a=0$ for $a<0$.
	\begin{enumerate}[(i)]
		\item Let $\lambda<0$. Suppose that $r_\lambda$ is convex and let $\mu < \lambda$. Let $x_1,x_2\in \X$ and $\eta \in[0,1]$. For every $t\in [0,+\infty]$, we have
		\[
		k_{\mu,\lambda}^\prime(t) = \frac{\mu}{\lambda} t^{\frac{\mu}{\lambda} -1} \geq 0,\qquad 
		k_{\mu,\lambda}^{\prime\prime}(t)= \Big(\frac{\mu}{\lambda}\Big)\Big(\frac{\mu}{\lambda} -1\Big) t^{\frac{\mu}{\lambda}-2} \geq 0
		\]
		since $\frac{\mu}{\lambda} > 1$. Hence $k_{\mu,\lambda}$ is convex and increasing. Observe that
		\begin{align*}
			r_{\mu}(\eta x_{1} + (1-\eta)x_{2}) 
			&= k_{\mu,\lambda}\circ r_{\lambda}(\eta x_{1} + (1-\eta)x_{2})\\
			&\leq k_{\mu,\lambda}(\eta r_{\lambda}( x_{1}) + (1-\eta)r_{\lambda}(x_{2}))\\
			&\leq \eta k_{\mu,\lambda}\circ r_{\lambda}(x_{1}) + (1-\eta)k_{\mu,\lambda}\circ r_{\lambda}(x_{2})
		    = \eta r_{\mu}(x_{1}) +(1-\eta) r_{\mu}(x_{2}),
		\end{align*}
		where the first inequality holds since $r_{\lambda}$ is convex, $k_{\mu,\lambda}$ is increasing and the second inequality follows from the convexity of $k_{\mu,\lambda}$. Therefore, $r_{\mu}$ is convex.
		
		Conversely, assume that $r_\mu$ is convex for each $\mu<\lambda$. Let $x_1,x_2\in \X$ and $\eta \in [0,1]$. For every $\mu<\lambda$, we have
		\[
					r_{\mu}(\eta x_{1} + (1-\eta)x_{2}) \leq \eta r_{\mu}(x_{1}) + (1- \eta) r_{\mu}(x_{2}).
		\]
		Thanks to the continuity of the power function $\mu\mapsto e^{-\mu a}$ on $(-\infty,0)$ for each fixed $a\in\bar{\R}$, we may let $\mu\rightarrow\lambda$ and get
		\[
			r_{\lambda}(\eta x_{1} + (1-\eta)x_{2}) \leq \eta r_{\lambda}(x_{1}) + (1- \eta)r_{\lambda}(x_{2}).
		\] 			
		Hence $r_{\lambda}$ is convex.
		
		\item Let $\lambda>0$. Suppose that $r_{\lambda}$ is concave. Clearly, $r_0\equiv 1$ is convex. Let $\mu\in (0,\lambda)$. Let $x_1,x_2\in \X$ and $\eta\in [0,1]$. Similar to (i), it is easy to check that $k_{\mu,\lambda}$ is concave and increasing since $0\leq \frac{\mu}{\lambda} < 1$. Observe that
		\begin{align*}
			r_{\mu}(\eta x_{1} + (1-\eta)x_{2}) 
			&= k_{\mu,\lambda}\circ r_{\lambda}(\eta x_{1} + (1-\eta)x_{2})\\
			&\geq k_{\mu,\lambda}(\eta r_{\lambda}( x_{1}) + (1-\eta)r_{\lambda}(x_{2}))\\
			&\geq \eta k_{\mu,\lambda}\circ r_{\lambda}(x_{1}) + (1-\eta)k_{\mu,\lambda}\circ r_{\lambda}(x_{2})
			= \eta r_{\mu}(x_{1}) +(1-\eta) r_{\mu}(x_{2}),
		\end{align*}
		where the first inequality holds since $r_{\lambda}$ is concave, $k_{\mu,\lambda}$ is increasing and the second inequality follows from concavity of $k$. Therefore, $r_{\mu}$ is concave.
		
		Conversely, assume that $r_\mu$ is concave for each $\mu\in [0,\lambda)$. Let $x_{1}, x_{2}\in \X$ and $\eta \in [0,1]$. For every $\mu\in[0,\lambda)$, we have
		\[
			r_{\mu}(\eta x_{1} + (1-\eta)x_{2}) \geq \eta r_{\mu}(x_{1}) + (1- \eta) r_{\mu}(x_{2}).
		\]
		By letting $\mu\to\lambda$ similar to (i), we get
		\[
			r_{\lambda}(\eta x_{1} + (1-\eta)x_{2}) \geq \eta r_{\lambda}(x_{1}) + (1- \eta)r_{\lambda}(x_{2}).
		\]
		Hence $r_{\lambda}$ is concave.
		
		\item Assume that there exists $\lambda>0$ such that $r_{\lambda}$ is concave. Let $\mu < 0$, $x_1,x_2\in \X$ and $\eta\in[0,1]$. Similar to (i) and (ii), we may conclude that $k_{\mu,\lambda}$ is convex and decreasing since $\frac{\mu}{\lambda} < 0$. Observe that
		\begin{align*}
			r_{\mu}(\eta x_{1} + (1-\eta)x_{2}) 
			&= k_{\mu,\lambda}\circ r_{\lambda}(\eta x_{1} + (1-\eta)x_{2})\\
			&\leq k_{\mu,\lambda}(\eta r_{\lambda}( x_{1}) + (1-\eta)r_{\lambda}(x_{2}))\\
			&\leq \eta k_{\mu,\lambda}\circ r_{\lambda}(x_{1}) + (1-\eta)k_{\mu,\lambda}\circ r_{\lambda}(x_{2})
			= \eta r_{\mu}(x_{1}) +(1-\eta) r_{\mu}(x_{2}),
		\end{align*}
		where the first inequality holds since $r_{\lambda}$ is concave, $k_{\mu,\lambda}$ is decreasing and the second inequality follows from convexity of $k$. Therefore, $r_{\mu}$ is convex.
	\end{enumerate}
\end{proof}

Let us consider the following two cases. First, assume that there exists $\lambda < 0$ such that $r_{\lambda}$ is not convex. Then, according to Lemma \ref{lemma 1 section 2}(i), $r_{\gamma}$ is not convex for every $\gamma\in[\lambda, 0)$. Moreover, if there exists $\mu < \lambda$ such that $r_{\mu}$ is convex, then $r_{\gamma}$ is convex for every $\gamma\in (-\infty,\mu]$. Second, assume otherwise that $r_{\lambda}$ is convex for every $\lambda < 0$. In view of Lemma \ref{lemma 1 section 2}(iii), this is a necessary condition for $r_{\mu}$ to be concave for some $\mu > 0$. Moreover, if $r_{\mu}$ is concave for some $\mu > 0$, then $r_{\gamma}$ is concave for every $\gamma\in [0,\mu)$ by Lemma \ref{lemma 1 section 2}(ii). Figure \ref{fig1} and Figure \ref{fig2} below depict these two cases.

\begin{figure}[h!]
	\centering
	\begin{tikzpicture}[scale= 1]
		
		\draw[thick]  (0,0)--(10,0) ;
		
		\draw (0,-.08)--(0,.08);
		\draw (2,-.08)--(2,.08);
		\draw (3,-.08)--(3,.08);
		\draw (5,-.08)--(5,.08);
		\draw (10,-.08)--(10,.08);
		
		\node [below] at (0,0) {$-\infty$};
		\node [below] at (10,0) {$+\infty$};
		\node [below] at (5,0) {$0$};
		
		\node [below] at (3,0) {$\lambda$};
		\node [below] at (2,0) {$\mu$};
		
		\draw [decorate,decoration={brace,amplitude=10pt}] (3,0.2) -- (5,0.2) ;
		\node [above] at (4,.8){not convex};
		\draw [decorate,decoration={brace,amplitude=10pt}] (0,0.2) -- (2,0.2) ;
		\node [above] at (1,.8){convex};
	\end{tikzpicture}
	\caption{$r_{\lambda}$ is not convex for some $\lambda < 0$.}
	\label{fig1}
\end{figure}
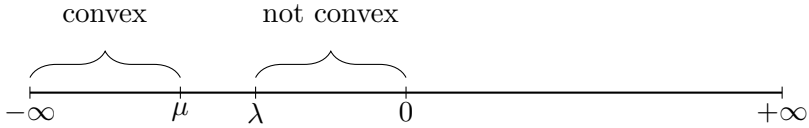

\begin{figure}[h!]
	\centering
	\begin{tikzpicture}[scale= 1]
		
		\draw[thick]  (0,0)--(10,0) ;
		
		\draw (0,-.08)--(0,.08);
		\draw (5,-.08)--(5,.08);
		\draw (8,-.08)--(8,.08);
		\draw (10,-.08)--(10,.08);
		
		\node [below] at (0,0) {$-\infty$};
		\node [below] at (10,0) {$+\infty$};
		\node [below] at (5,0) {$0$};
		
		\node [below] at (8,0) {$\mu$};

		\draw [decorate,decoration={brace,amplitude=10pt}] (5,0.2) -- (8,0.2) ;
		\node [above] at (6.5,.8){concave};
		
		\draw [decorate,decoration={brace,amplitude=10pt}] (0,0.2) -- (5,0.2) ;
		\node [above] at (2.5,.8){convex};
	\end{tikzpicture}
	\caption{$r_{\lambda}$ is convex for every $\lambda < 0$.}
	\label{fig2}
\end{figure}
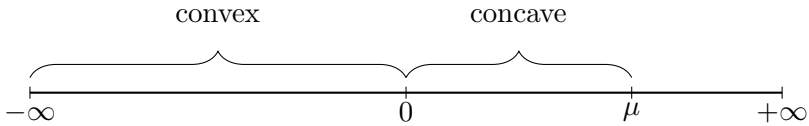

Considering the cases described above, an intriguing query arises as to determine the largest $\lambda<0$ for which $r_{\lambda}$ is convex (first case) and the largest $\lambda\geq 0$ for which $r_{\lambda}$ is concave (second case). The following definition is motivated by these ideas.

\begin{defn} \label{convindex_def_infdimensional}
	The convexity index $c(f)\in\bar{\R}$ of $f$ is defined as follows:
	\begin{enumerate}[(i)]
		\item if there exists $\bar{\lambda}<0$ such that $r_{\bar{\lambda}}$ is not convex, then
		\[
			c(f) \coloneqq \sup\{\lambda<0 \colon  r_{\lambda} \text{ is convex} \}.
		\]
		\item if $r_{\bar{\lambda}}$ is convex for every $\bar{\lambda}<0$, then 
		\[
			c(f) \coloneqq \sup\{\lambda \geq 0 \colon r_{\lambda} \text{ is concave} \}.
		\]
	\end{enumerate}
\end{defn}

\begin{rem}\label{cf-remark}
	By the discussion preceding Definition \ref{convindex_def_infdimensional}, it is clear that $c(f)\in [-\infty,0)$ in case (i) and $c(f)\in [0,+\infty]$ in case (ii).
\end{rem}

As its name suggests, the convexity index of $f$ tells whether $f$ is convex or not as stated in the next theorem.

\begin{thm}\label{convindex_values_infinitedimensional}
	The function $f$ is convex if and only if $c(f) \geq 0$.
\end{thm}
\begin{proof}
	Assume that $f$ is a convex function. Let $\bar{\lambda} < 0$. For every $\eta \in [0,1]$, $x_{1}, x_{2}\in \X$,
	\[
		r_{\bar{\lambda}}(\eta x_{1} + (1-\eta)x_{2})= e^{-\bar{\lambda} f(\eta x_{1} + (1-\eta)x_{2})}
\leq e^{-\bar{\lambda} (\eta f(x_{1}) + (1-\eta)f(x_{2}))}
\leq \eta r_{\bar{\lambda}}(x_{1}) + (1-\eta) r_{\bar{\lambda}}(x_{2}),
	\]
	where the first inequality follows from the convexity of $f$ and monotonicity of $t\mapsto e^{-\overline{\lambda} t}$ , the second inequality holds since $t\mapsto e^{-\overline{\lambda} t}$ defined on $[0,+\infty]$ is convex. Therefore,
	$r_{\bar{\lambda}}$ is convex. Hence, by Definition \ref{convindex_def_infdimensional},
	$c(f) = \sup \{\lambda\geq 0 \colon r_{\lambda}\text{ is concave} \}$ so that $c(f) \geq 0$.
	
	Conversely, assume that $c(f)\geq 0$. Let $x_{1}, x_{2}\in \X$ and $\eta\in [0,1]$. We claim that
	\begin{equation}\label{jensen}
		f(\eta x_{1} + (1-\eta )x_{2}) \leq \eta f(x_{1}) + (1-\eta)f(x_{2}).
	\end{equation}
	Note that \eqref{jensen} holds trivially if $f(x_{1})=+\infty$ or $f(x_{2})=+\infty$. Hence, we consider the following cases:
	\begin{enumerate}[(a)]
		\item $f(x_{1}) < +\infty$, $f(x_{2}) < +\infty$,  $f(\eta x_{1} + (1-\eta)x_{2}) < +\infty$,
		\item $f(x_{1}) < +\infty$, $f(x_{2}) < +\infty$,  $f(\eta x_{1} + (1-\eta)x_{2}) = +\infty$.
	\end{enumerate} 
	For each $\lambda\in \R$, define
	\begin{align*}
		k(\lambda)
		&\coloneqq  r_{\lambda}(\eta x_{1} + (1-\eta)x_{2}) - \eta r_{\lambda}(x_{1})- (1-\eta) r_{\lambda}(x_{2})\\
		&= e^{-\lambda f(\eta x_{1} + (1-\eta)x_{2})} - \eta e^{-\lambda f(x_{1})} - (1-\eta) e^{-\lambda f(x_{2})}.
	\end{align*}
	By Remark \ref{cf-remark}, $r_{\lambda}$ is convex for every $\lambda<0$, which implies that $k(\lambda)\leq 0$ for every $\lambda < 0$. Observe that $k(\lambda) = +\infty$ for every $\lambda < 0$ if case (b) is true, which is in contradiction with the previous statement. Therefore, we are only left with case (a). Noting that $k(0)=0$, we get
	\[
		k^\prime_{-}(0) \coloneqq \lim_{\lambda\rightarrow 0^{-}}\frac{k(\lambda) - k(0)}{\lambda} \geq 0.
	\]
	Furthermore, $k$ is differentiable everywhere on $\R$, in particular, at $\lambda=0$. Hence, the derivative and the left derivative of $k$ at 0 are equal, that is,
	\[
	0 \leq k^\prime_{-}(0) =  k^\prime(0) = -f(\eta x_{1} + (1-\eta )x_{2}) + \eta f(x_{1}) + (1-\eta)f(x_{2}).
	\]
	Therefore, \eqref{jensen} holds. Since $\eta, x_{1}, x_{2}$ are arbitrary, $f$ is convex. 
\end{proof}

Before proceeding further, we give a basic scaling property that will be useful later.
\begin{lem} \label{Homegenity Convexity Index}
	Let $w\in\R_{+}$. Then, $c(wf) = \frac{1}{w}c(f)$.
\end{lem}
\begin{proof}
	The result follows trivially when $w = 0$. Suppose $w>0$. We consider the following cases:
	\begin{enumerate}[(a)]
		\item Suppose that $f$ is convex. Then, $wf$ is convex. Theorem \ref{convindex_values_infinitedimensional} together with Remark \ref{cf-remark} implies that $c(wf) = \sup\{\lambda\in\R \colon \lambda\geq 0, e^{-\lambda w f} \text{ is concave}\}$. Observe that
		\begin{align*}
			c(w f) 
			&= \sup\{\lambda\in\R \colon \lambda\geq 0, e^{-\lambda w f} \text{ is concave}\}\\
			&= \frac{1}{w}\sup\{\lambda w\in\R \colon \lambda\geq 0, e^{-\lambda w f} \text{ is concave}\}\\
			&=\frac{1}{w}\sup\{\mu\geq 0 \colon e^{-\mu f} \text{ is concave}\}= \frac{1}{w}c(f).
		\end{align*}
		
		\item Suppose that $f$ is not convex. Then, $wf$ is not convex. Theorem \ref{convindex_values_infinitedimensional} together with Remark \ref{cf-remark} implies that  $c(w f) = \sup\{\lambda\in\R \colon \lambda < 0, e^{-\lambda w f} \text{ is convex}\}$. Similar to case (a), we have
		\begin{align*}
			c(w f) 
			&= \sup\{\lambda\in\R \colon \lambda < 0, e^{-\lambda w f} \text{ is convex}\}\\
			&= \frac{1}{w}\sup\{\lambda w\in\R \colon \lambda< 0, e^{-\lambda w f} \text{ is convex}\}\\
			&=\frac{1}{w}\sup\{\mu<0 \colon e^{-\mu f} \text{ is convex}\}= \frac{1}{w}c(f).
		\end{align*}
	\end{enumerate}
	Hence, the result holds. 
\end{proof}

As discussed in \citet[Proposition 3(b)]{Crouzeix}, constant real-valued functions defined on $\R^{n}$ can be identified by their convexity indices. In the proof of this result, they use the fact that a convex function with finite values is continuous. Such a result does not hold for extended real-valued functions defined on general topological vector spaces. Fortunately, the characterization is still valid under a mild continuity condition, as shown in the next theorem.

\begin{thm}\label{convindex_values_infinitedimensional_constant}
	Assume that $f$ is lower semicontinuous. Then, $f$ is a constant function if and only if $c(f) = +\infty$.
\end{thm}
\begin{proof}
	Assume that $f$ is a constant function. Then, for every $\lambda\in\R$, the function $r_{\lambda}$ is constant, hence convex. By Definition \ref{convindex_def_infdimensional}, we have
	\begin{equation}\label{cfcase2}
		c(f) = \sup\{\lambda \geq 0\colon r_{\lambda} \text{ is concave} \}.
	\end{equation}
	Moreover, for each $\lambda\geq 0$, the constant function $r_\lambda$ is also concave. Therefore, $c(f) = +\infty$.
	
	Conversely, assume that $c(f) = +\infty$. Then, by Theorem \ref{convindex_values_infinitedimensional}, $f$ is convex. To get a contradiction, suppose that $f$ is not constant. We need to consider the following two cases:
	\begin{enumerate}[(a)]
		\item Suppose that $f$ takes only one finite value, that is, $f(\X)=\cb{c,+\infty}$ for some $c\in\R$. Then, for every $\lambda<0$, we have $r_\lambda(x)=e^{-\lambda c}$ for every $x\in\dom f$ and $r_\lambda(x)=+\infty$ for every $x\in \X\sm\dom f$; hence $r_\lambda$ is convex. So \eqref{cfcase2} is valid. Moreover, since $f$ is proper and lower semicontinuous, $\dom f = \cb{x\in \X\colon f(x)\leq c}$ is a closed set. Let $\lambda>0$. Then, $r_\lambda(x)=e^{-\lambda c}>0$ for every $x\in\dom f$ and $r_\lambda(x)=0$ for every $x\in \X\sm\dom f$. Let $x_1\in \dom f$ and $x_2\in X\sm \dom f$. Since $\X$ is a topological vector space, the function $[0,1]\ni \eta\mapsto x^\eta\coloneqq \eta x_1+(1-\eta)x_2\in \X$ is continuous. Hence, $\lim_{\eta\rightarrow 0}x^\eta=x_2$. Since $x_2\in \X\sm \dom f$ and $\X\sm\dom f$ is an open set, continuity at $\eta =0$ implies that there exists $\bar{\eta}\in(0,1)$ such that $x^{\bar{\eta}}\in \X\sm\dom f$. In particular, $f(x^{\bar{\eta}})=+\infty$ and $r_\lambda(x^{\bar{\eta}})=0$. It follows that
		$r_\lambda(x^{\bar{\eta}})=0< \bar{\eta}e^{-\lambda c} =\bar{\eta}r_\lambda(x_1)+(1-\bar{\eta})r_\lambda(x_2)$	so that $r_\lambda$ is not concave. Therefore, $c(f)=0$ by \eqref{cfcase2}, which is a contradiction to $c(f)=+\infty$. So this case is eliminated.
		\item Suppose that $f$ takes at least two finite values, that is, there exist $x_{1}, x_{2}\in \X$ such that $f(x_{1}) < f(x_{2})<+\infty$. Since $f$ is convex and $f(x_{1}) < f(x_{2})$, for every $\eta \in (0,1)$, we have
		\begin{equation}\label{eq:Theorem Constant 1} 
			f(\eta x_{1} + (1-\eta)x_{2}) \leq \eta f(x_{1}) + (1-\eta) f(x_{2}) < f(x_{2}). 
		\end{equation}
		Moreover, since $f$ is lower semicontinuous at $x_2$, we have $ f(x_2)\leq \liminf_{x\to x_{2}} f(x)$,
		where
		\[
		\liminf_{x\to x_{2}} f(x) = \sup \cb{ \inf\cb{f(x) \colon x\in U \setminus \{x_{2}\}  } \colon U\subseteq \X\text{ is open, } x_{2}\in U, U \setminus\{x_{2}\}\neq \emptyset \}  }.
		\]
		Since $f(x_1)<f(x_2)\leq \liminf_{x\to x_{2}} f(x)$, there exists an open neighborhood $\bar{U}$ of $x_{2}$ such that $f(x_1)<\inf\{f(x) \colon x\in \bar{U} \setminus \{x_{2}\} \}$. Observe that 
		\begin{equation}\label{eq:Theorem Constant 2}
			f(x_1) < f(x)\; \text{for every } x\in \bar{U}.
		\end{equation}
		Similar to case (a), the function $ [0,1]\ni \eta\mapsto x^\eta\coloneqq \eta x_{1} + (1-\eta) x_{2}\in \X$ is continuous with $\lim_{\eta \to 0} x^\eta = x_{2}$, which implies that there exists $\bar{\eta}\in(0,1)$ such that $x^{\bar{\eta}} \in \bar{U}$. By \eqref{eq:Theorem Constant 2}, $f(x_1)<f(x^{\bar{\eta}})$. Together with \eqref{eq:Theorem Constant 1}, we have
		\begin{equation}\label{eq:Theorem Constant 3}
			f(x_{1}) < f(x^{\bar{\eta}}) < f(x_{2}).
		\end{equation}
		By Remark \ref{cf-remark}, $r_{\lambda}$ is concave for every $\lambda \geq 0$ since $c(f) = + \infty$. Therefore,
		\begin{equation} \label{eq:Theorem Constant 4}
			r_{\lambda}(x^{\bar{\eta}}) - \bar{\eta} r_{\lambda}(x_{1}) - (1-\bar{\eta}) r_{\lambda}(x_{2}) \geq 0 \; \text{for every }\lambda > 0.
		\end{equation}
		Let $\lambda>0$. With some algebraic operations, \eqref{eq:Theorem Constant 4} can be rewritten as
		\begin{equation}\label{eq:Theorem Constant 5} 
			\frac{ r_{\lambda}(x_{1}) - r_{\lambda}(x_{2}) }{ r_{\lambda}(x^{\bar{\eta}}) - r_{\lambda}(x_{2})} \leq \frac{1}{\bar{\eta}}.  
		\end{equation}
		Note that
		\begin{align*}
			\frac{ r_{\lambda}(x_{1}) - r_{\lambda}(x_{2}) }{ r_{\lambda}(x^{\bar{\eta}}) - r_{\lambda}(x_{2})}&= \frac{ r_{\lambda}(x_{1})  -  r_{\lambda}(x^{\bar{\eta}}) +  r_{\lambda}(x^{\bar{\eta}}) - r_{\lambda}(x_{2})}{ r_{\lambda}(x^{\bar{\eta}}) -  r_{\lambda}(x_{2}) }
			= \frac{r_{\lambda}(x_{1}) -  r_{\lambda}(x^{\bar{\eta}})}{ r_{\lambda}(x^{\bar{\eta}}) -  r_{\lambda}(x_{2}) } + 1\\
			&\geq  \frac{r_{\lambda}(x_{1}) -  r_{\lambda}(x^{\bar{\eta}})}{ r_{\lambda}(x^{\bar{\eta}}) -  r_{\lambda}(x_{2}) }\geq   \frac{r_{\lambda}(x_{1}) -  r_{\lambda}(x^{\bar{\eta}})}{ r_{\lambda}(x^{\bar{\eta}}) }
			\geq  \frac{r_{\lambda}(x_{1})}{ r_{\lambda}(x^{\bar{\eta}}) } - 1= e^{-\lambda (f(x_{1})- f(x^{\bar{\eta}}))} - 1.
		\end{align*}
		By \eqref{eq:Theorem Constant 3}, $f(x_{1})- f(x^{\bar{\eta}})< 0$. Therefore,
		\[
			\limsup_{\lambda\to\infty} \frac{ r_{\lambda}(x_{1}) - r_{\lambda}(x_{2}) }{ r_{\lambda}(x^{\bar{\eta}}) - r_{\lambda}(x_{2})}
			\geq \limsup_{\lambda\to\infty} e^{-\lambda (f(x_{1})- f(x^{\bar{\eta}}))} - 1
			= \lim_{\lambda\to\infty}e^{-\lambda (f(x_{1})- f(x^{\bar{\eta}}))} - 1
			= + \infty.
		\]
		This implies that there exists $\bar{\lambda} > 0$ such that
		\[
		\frac{ r_{\bar{\lambda}}(x_{1}) - r_{\bar{\lambda}}(x_{2}) }{ r_{\bar{\lambda}}(x^{\bar{\eta}}) - r_{\bar{\lambda}}(x_{2})} > \frac{1}{\bar{\eta}}.
		\]
		By \eqref{eq:Theorem Constant 5}, $r_{\bar{\lambda}}$ is not concave, which is in contradiction with \eqref{eq:Theorem Constant 4}. Hence, $f$ is constant.
	\end{enumerate}		
\end{proof}	 

\subsection{Finite decomposable sums}\label{subsec:finsum}

In this section, we study additively decomposable quasiconvex functions on general topological vector spaces. The main results are (i) Theorem \ref{TwoFactorCase}, which characterizes a quasiconvex decomposable sum in terms of the sum of the convexity indices of its coordinate functions, and (ii) Theorem \ref{MainResult-2.2}, in which we consider a quasiconvex decomposable sum, and show that either all the coordinate functions are convex, or all except one are convex and convexity indices of coordinate functions satisfies an additive formula, and vice versa.

First we prove a technical result that we need in the sequel.

\begin{lem} \label{Lemma 2.2 - 1}
	Suppose that $f$ is quasiconvex and not convex. Then, there exist $x_{1},x_{2}\in \dom f$, $\bar{t}\in[0,1)$ and $\alpha\in\R$ satisying the following properties:
	\begin{enumerate}[(i)]
	\item It holds $f(x_{2}) < \alpha \leq f(x_{1})$.
	\item For every $t\in [0, \bar{t}]$, it holds $f(tx_{1} + (1-t) x_{2} ) \leq \alpha + (f(x_{1}) - f(x_{2})) (t - \bar{t})$.
	\item For every $t\in (\bar{t},1]$, it holds $\alpha \leq f(tx_{1} + (1-t) x_{2} ) < \alpha + (f(x_{1}) - f(x_{2})) (t - \bar{t})$.
	\end{enumerate} 
\end{lem}
\begin{proof}
	Since $f$ is not convex, there exist $x_{1}, x_{2} \in \dom f$ such that $f(x_1)\geq f(x_2)$ and
	\begin{equation}\label{eq:Delta1}
		\Delta \coloneqq \sup_{t\in[0,1]}  g(t) > 0, 
	\end{equation}
	where $g(t) \coloneqq f( t x_{1} + (1-t) x_{2} ) - t f(x_{1}) - (1-t) f(x_{2})$. 
	Then, for every $ t\in[0,1]$, 
	\[
		g(t) \leq f(tx_{1} + (1-t)x_{2} ) - f(x_{2}) \leq f(x_{1}) - f(x_{2}),
	\]
	where the first inequality holds since $f(x_{1}) \geq f(x_{2})$ and the second inequality follows from the quasiconvexity of $f$. Hence, 
	\begin{equation} \label{eq:Delta2} 
		0 < \Delta \leq f(x_{1}) - f(x_{2}).  
	\end{equation}
	Let $\gamma\in[0,\Delta)$ and define $t_{\gamma} \coloneqq \sup\{  t\in[0,1] \colon g(t) \geq \gamma\}$. Obviously, $t_{\gamma} > 0$. We claim that
	\begin{equation}\label{eq:Delta4}
		t_{\gamma} \leq \frac{f(x_{1}) - f(x_{2}) -\gamma}{f(x_{1}) - f(x_{2})}.
	\end{equation}
	Indeed, supposing otherwise, there exists $t^{*} \in \Big( \frac{f(x_{1}) - f(x_{2}) -\gamma}{f(x_{1}) - f(x_{2})}, t_{\gamma} \Big]$ such that 
	\begin{equation}\label{eq:Delta3}
		g(t^{*}) \geq \gamma.
	\end{equation}
	Observe that
	\[
		t^* > \frac{f(x_{1}) - f(x_{2}) -\gamma}{f(x_{1}) - f(x_{2})} \quad \iff \quad  t^* f(x_1) +(1-t^*)f(x_2) >  f(x_{1}) - \gamma.
	\]
	Hence,
	\begin{align*}
		g(t^{*}) 
		&= f(t^{*} x_{1} + (1-t^{*}) x_{2} ) - t^{*}f(x_{1}) - (1-t^{*})f(x_{2})\\
		&< f(t^{*}x_{1} + (1-t^{*})x_{2}) - f(x_{1}) +\gamma
		\leq \max \{f(x_{1}) , f(x_{2}) \} -f(x_{1}) + \gamma
		= \gamma,
	\end{align*}
	where the second inequality follows from the quasiconvexity of $f$ while the last equality follows from the assumption that $f(x_{1}) \geq f(x_{2})$. This is in contradiction with \eqref{eq:Delta3}. Therefore,
	\eqref{eq:Delta4} holds for every $\gamma\in[0,\Delta)$. 

	Note that $\gamma\mapsto t_\gamma$ is a nonincreasing function.	Let us define
	\[
	\overline{t} \coloneqq \lim_{\gamma \to \Delta} t_{\gamma}, \quad\quad \alpha \coloneqq \Delta + \bar{t}f(x_{1}) + (1-\bar{t}) f(x_{2}).
	\]
	It follows from \eqref{eq:Delta2} and letting $\gamma\to\Delta$ in \eqref{eq:Delta4} that 
	\begin{equation} \label{eq:Delta6}
		0\leq \overline{t} \leq \frac{f(x_{1}) - f(x_{2}) -\Delta}{f(x_{1}) - f(x_{2})} < 1. 
	\end{equation}
	Therefore, 
	\[
		f(x_{2})\negthinspace <\negthinspace f(x_{2}) \negthinspace+ \negthinspace\Delta \negthinspace+\negthinspace \overline{t}(f(x_{1}) - f(x_{2})) 
		= \alpha
		\leq f(x_{2}) \negthinspace+\negthinspace \Delta \negthinspace+\negthinspace \frac{f(x_{1}) - f(x_{2}) -\Delta}{f(x_{1}) - f(x_{2})}(f(x_{1}) - f(x_{2}))= f(x_{1}),
	\]
	where the strict inequality holds since $\Delta + \overline{t}(f(x_{1}) - f(x_{2}))>0$ by \eqref{eq:Delta2}, \eqref{eq:Delta6}; and the non-strict inequality follows from \eqref{eq:Delta6}. Hence, property (i) follows.
	
	By the definition of $\Delta$, for every $t\in[0,1]$,
	\begin{equation}\label{eq:Delta7}
		f(t x_{1} + (1-t) x_{2}) \leq \Delta + t f(x_{1}) + (1-t) f(x_{2})  = \alpha + (t-\overline{t}) (f(x_{1}) - f(x_{2})).
	\end{equation}
	In particular, property (ii) holds. Now, assume that the inequality above holds with equality for some $\hat{t} \in[0,1]$. Then, $g(\hat{t}) = \Delta$. Observe that $g(\hat{t}) > \gamma$ for every $\gamma \in [0,\Delta)$. By the monotonicity of $\gamma\mapsto t_\gamma$, we have $\hat{t} \leq t_{\gamma}$ for every $\gamma \in [0,\Delta)$. Letting $\gamma\to\Delta$, we have $\hat{t} \leq \overline{t}$. Thus, letting $t\in (\overline{t}, 1]$, we have
	\begin{equation}\label{eq:Delta8}
		f(t x_{1} + (1-t) x_{2}) < \alpha +  (t-\overline{t}) (f(x_{1}) - f(x_{2})).
	\end{equation}
	It remains to show that $\alpha \leq f(t x_{1} + (1-t) x_{2})$. By the definition of $\overline{t}$ and the monotonicity $\gamma\mapsto t_{\gamma}$, we have $t_{\gamma} \geq \overline{t}$ for every $\gamma\in [0,\Delta)$. Hence, there exists $\overline{\gamma}\in (0, \Delta)$ such that
	\begin{equation} \label{eq:Delta9}
		\overline{t}\leq t_{\gamma} < t
	\end{equation}
	for every $\gamma\in [\overline{\gamma}, \Delta)$. Moreover, for every $\gamma\in [\overline{\gamma}, \Delta)$ and $\epsilon\in \big(0, \frac{t_{\gamma}}{2}\big)$, there exists $t_{\gamma, \epsilon} \in [t_{\gamma} - \epsilon, t_{\gamma}]$ such that $g(t_{\gamma, \epsilon}) \geq \gamma$. By the definition of $g$, this is equivalent to
	\begin{equation} \label{eq:Delta10}
		f(t_{\gamma, \epsilon} x_{1} + (1-t_{\gamma, \epsilon})x_{2}) \geq \gamma + f(x_{2}) + t_{\gamma, \epsilon} (f(x_{1}) - f(x_{2}) ) > f(x_{2}). 
	\end{equation}
	Note that $t \mapsto f(tx_{1} + (1-t) x_{2})$ is quasiconvex on $[0,1]$ as it is the composition of an affine function with a quasiconvex function. Observe that $t_{\gamma,\epsilon}$ can be written as a convex combination of $0$ and $t$ since $0< t_{\gamma, \epsilon}\leq t_{\gamma} < t$. Therefore, 
	$f(t_{\gamma, \epsilon} x_{1} + (1-t_{\gamma, \epsilon})x_{2}) \leq \max\{f(x_{2}), f(t x_{1} + (1-t) x_{2})\}$.
	By (\ref{eq:Delta10}), we have $f(t_{\gamma, \epsilon} x_{1} + (1-t_{\gamma, \epsilon}) x_{2}) > f(x_{2})$. Hence, 
	\begin{equation} \label{eq:Delta11}
		f(t x_{1} + (1-t) x_{2}) \geq f(t_{\gamma, \epsilon} x_{1} + (1-t_{\gamma, \epsilon})x_{2}).
	\end{equation}
	It follows from \eqref{eq:Delta9}, \eqref{eq:Delta10} and \eqref{eq:Delta11} that
	\begin{align*}
		f(t x_{1} + (1-t)x_{2}) \geq \gamma + f(x_{2}) + t_{\gamma, \epsilon} (f(x_{1}) - f(x_{2})) 
		&\geq \gamma + f(x_{2} ) + (t_{\gamma} - \epsilon)( f(x_{1}) - f(x_{2}) )\\
		&\geq \gamma + f(x_{2}) + (\overline{t} - \epsilon) (f(x_{1}) - f(x_{2}) ) .
	\end{align*}
	Letting $\epsilon \to 0$ and $\gamma\to\Delta$ gives
	\begin{equation} \label{eq:DeltaResult3}
		f(t x_{1} + (1-t) x_{2}) \geq \Delta + f(x_{2}) + \bar{t}(f(x_{1})- f(x_{2})) = \alpha.
	\end{equation}
	Hence, property (iii) follows.
\end{proof}

For the rest of this section, we fix $n\in\N$; for each $i \in \{1, \ldots, n \}$, we let $\X_{i}$ be a topological vector space and $f_{i}$ a proper extended real-valued function on $\X_{i}$. Furthermore, we define a function $s \colon \X_{1}\times \ldots \times \X_{n} \to \R \cup \{ + \infty \}$ by
\begin{equation}\label{Section 2 - Sum Function}
	s(x_{1}, \ldots, x_{n}) \coloneqq f_{1}(x_{1}) + \ldots + f_{n}(x_{n}).
\end{equation}
In \citet{Debreu}, and \citet{Crouzeix}, the implications of the quasiconvexity of $s$ on $f_{i}$ are studied when $\X_{i}$ is an open subset of $\R^{n}$. We next extend their results to general topological vector spaces. 
\begin{thm}\label{TwoFactorCase}
	Assume that $f_{1}, \ldots, f_{n}$ are non-constant. Then, $s$ is quasiconvex if and only if 
	\begin{align*}
		c(f_{1}) + \ldots + c(f_{n}) \geq 0.
	\end{align*}
\end{thm}
\begin{proof}
	It is enough to consider the case $n= 2$ and we write $f = f_{1}$, $g = f_{2}$, $\X_{1} = \X$ and $\X_{2} = \Y$. First, assume that $s$ is quasiconvex and suppose to the contrary that $c(f)+c(g) < 0$. Given $\overline{y} \in \dom g$, for every $x_{1}, x_{2}\in  \X$, and $t \in [0,1]$,
	\begin{align*}
		f(t x_{1} + (1-t) x_{2}) + g(\overline{y})			 
		&= s( t x_{1}+ (1-t) x_{2},\overline{y})\\
		&\leq \max\{s(x_{1}, \overline{y}), s(x_{2},\overline{y})\}
		= \max\{f(x_{1}), f(x_{2})\} + g(\overline{y}).
	\end{align*}
	Subtracting $g(\overline{y})$ from both sides yields $f( t x_{1} + (1-t)x_{2}) \leq \max \{ f(x_{1}), f(x_{2})\}$. Therefore, $f$ is quasiconvex. By symmetry, $g$ is also quasiconvex. Without loss of generality we can assume that $c(f) \leq c(g)$. Then, there exists $\lambda < 0$ such that
	\begin{equation} \label{eq:TwoFactorCase 2}
		c(f) < \lambda < -c(g).
	\end{equation}
	For every $(x,y)\in \X \times \Y$, define
	\[
		\overline{f}(x) \coloneqq  e^{-\lambda f(x)},\quad \overline{g}(y) \coloneqq e^{\lambda g(y)},\quad \overline{s}(x,y) \coloneqq e^{-\lambda s(x,y)}.
	\]
	Note that $\overline{s}$ is a quasiconvex function since it is the composition of a quasiconvex function with a non-decreasing function.  It follows from Definition \ref{convindex_def_infdimensional} and Remark \ref{cf-remark} that $\overline{f}$ is not convex.  Moreover, $-\overline{g}$ is not convex. To see this, assume to the contrary that $\overline{g}$ is concave. Then, $r_{\mu}$ associated with $g$ is convex for every $\mu < 0$ by Lemma \ref{lemma 1 section 2}-(iii). Hence, $c(g) = \sup\{\gamma \in\R \colon \gamma \geq 0, r_{\gamma} \text{ is concave} \}$ by Definition \ref{convindex_def_infdimensional}. Since $\overline{g}$ is assumed to be concave, $-\lambda \in \{\gamma \in\R \colon \gamma \geq 0, r_{\gamma} \text{ is concave} \}$. By Lemma \ref{lemma 1 section 2}, $ c(g) \geq - \lambda$. This is in contradiction with \eqref{eq:TwoFactorCase 2}.
	
	Observe that $\overline{f}$ and $-\overline{g}$ are compositions of $f$ and $g$ respectively with a non-decreasing function. Therefore $\overline{f}$ and $-\overline{g}$ are quasiconvex. 
	
	Next, we apply Lemma \ref{Lemma 2.2 - 1} to $\overline{f}$ and $-\overline{g}$ since they are quasiconvex but not convex functions. There exist $x_{1}, x_{2}\in \dom \overline{f}$, $y_{1}, y_{2}\in \dom (-\overline{g})$, $\overline{t}, \overline{u}\in [0,1)$ and $\alpha,\beta \in\R$ such that
	\begin{align}
		&0 < \theta (0) < \alpha \leq \theta(1), \quad &&\mu(0) < -\beta \leq \mu(1) < 0, \label{twofactors1}\\
		&\theta(t) \leq \alpha + m(t - \overline{t}),\quad t\in [0,\overline{t}],  \quad &&\mu(u) \leq -\beta + n(u - \overline{u}),\quad u\in[0, \overline{u}],\label{twofactors2}\\
		&\alpha \leq \theta(t) < \alpha + m(t-\overline{t}),\quad t\in(\overline{t}, 1],\quad &&-\beta \leq \mu(u) < -\beta + n(u-\overline{u}),\quad u\in (\overline{u}, 1], \label{twofactors3}
	\end{align}	
	where
	\begin{align*}
		\theta(t) &\coloneqq \overline{f}(t x_{1}+ (1-t) x_{2}),\quad t\in[0,1], 
		&&\mu(t) \coloneqq -\overline{g}(u y_{1} + (1-u) y_{2}),\quad u\in[0,1],\\
		m &\coloneqq \theta(1) - \theta(0) > 0,\quad && n \coloneqq \mu(1) - \mu(0) > 0.
	\end{align*}
	Let us also define
	\begin{align*}
		&\xi (t,u) \coloneqq -\frac{\theta(t)}{\mu(u)},\quad (t,u)\in [0,1]\times[0,1],\\
		&S \coloneqq \Big\{(t,u)\in [0,1]\times [0,1]\colon \xi (t,u) < \frac{\alpha}{\beta} \Big\},\;
		T \coloneqq \{(t,u)\in [0,1]\times [0,1]\colon m\beta (t - \overline{t}) + n\alpha (u -\overline{u}) \leq 0 \}.
	\end{align*}
	Notice that
	\[
		\xi(t,u) =e^{-\lambda s(t x_{1} + (1-t)x_{2} , u y_{1} + (1-u)y_{2})},\quad (t,u)\in [0,1]\times[0,1].
	\]
	Hence, $\xi$ is quasiconvex since $s$ is quasiconvex and $S$ is a convex set as the strict lower level set of a quasiconvex function. Moreover,
	\[
		S = \{(t,u)\in [0,1]\times [0,1]\colon  \beta \theta(t) + \alpha \mu(u) < 0 \}.
	\]
	From \eqref{twofactors2} and \eqref{twofactors3}, we immediately derive
	\begin{align}
		&S\cap \of{ (\overline{t},1]\times(\overline{u},1] } = \emptyset, \label{twofactors7}\\
		&S \supseteq T \cap \of{ [0,\overline{t}]\times(\overline{u},1]}, \label{twofactors8}\\
		&S \supseteq T \cap \of{ (\overline{t},1]\times[0,\overline{u}]},  \label{twofactors9}
	\end{align}
	There are four possible cases:
	\begin{enumerate}[(a)]
		\item Suppose that $\overline{t} > 0$ and $\overline{u} > 0$. Let
		\[
			k \coloneqq \frac{m \beta}{n \alpha},\qquad
			\epsilon \coloneqq \min \Big\{ \overline{t}, \frac{1 - \overline{u}}{k}, 1 - \overline{t}, \frac{\overline{u}}{k} \Big\}.
		\]
		It is not hard to see that $(\overline{t} - \epsilon, \overline{u} + k \epsilon)\in T \cap \of{ [0,\overline{t}]\times(\overline{u},1] }$ and $(\overline{t} + \epsilon, \overline{u} - k \epsilon) \in T \cap \of{ (\overline{t},1]\times[0,\overline{u}] }$. In view of \eqref{twofactors8} and \eqref{twofactors9}, $(\overline{t} - \epsilon, \overline{u} + k \epsilon), (\overline{t} + \epsilon, \overline{u} - k \epsilon) \in S$. It follows from the convexity of $S$ that
		\[
			(\overline{t}, \overline{u}) = \frac{1}{2} (\overline{t} - \epsilon, \overline{u} + k \epsilon) +  (1- \frac{1}{2}) (\overline{t} + \epsilon, \overline{u} - k \epsilon) \in S.
		\]
		Hence, $\beta\theta(\overline{t}) + \alpha \mu(\overline{u}) < 0$. On the other hand, by (\ref{twofactors2}), $\beta \theta(\overline{t}) \leq \alpha \beta$ and $\alpha \mu(\overline{u}) \leq -\alpha \beta$. Then, we have either $\alpha \mu(\overline{u}) < -\alpha \beta$ or $\beta \theta(\overline{t}) < \alpha \beta$. 
		
		First, assume that $\beta \theta(\overline{t}) < \alpha \beta$. Then, we may choose $u^{*}\in (\overline{u}, 1]$ such that 
		\[
		u^{*} - \overline{u} < \frac{\alpha \beta - \beta}{ n\alpha },
		\]
		that is, $n\alpha (u^{*} - \overline{u}) < \alpha \beta - \beta \theta(\overline{t})$. For instance, one can choose $u^{*} = \overline{u} + \frac{1}{2} \frac{\alpha \beta - \beta \theta(\overline{t})}{n \alpha}$. It follows from \eqref{twofactors3} that
		\[
			\beta \theta(\overline{t}) + \alpha \mu(u^{*}) < \beta \theta(\overline{t}) - \alpha \beta + n\alpha (u^{*} - \overline{u}) < 0.
		\]
		Therefore, $(\overline{t}, u^{*})\in S$. Note that one can choose $(t^{**},u^{**})\in (\overline{t},1] \times [0,\overline{u}]$ such that
		\[
		m\beta (t^{**} - \overline{t}) + n\alpha(u^{**}- \overline{u}) \leq 0.
		\]
		Hence, $(t^{**},u^{**})\in T \cap \of{(\overline{t},1] \times [0,\overline{u}]}$, and consequently $(t^{**},u^{**})\in S$ by \eqref{twofactors9}. Notice that, since $ \overline{t} < t^{**}$ and $u^{**} < \overline{u} < u^{*}$, there exists $\overline{\lambda}\in(0,1)$ such that
		\[
		\of{ \overline{\lambda}t^{**} + (1-\overline{\lambda})\overline{t},  \overline{\lambda} u^{**} + (1-\overline{\lambda})u^{*} } \in (\overline{t},1]\times(\overline{u},1].
		\]
		Moreover, $S$ is convex. Therefore, $( \overline{\lambda}t^{**} + (1-\overline{\lambda})\overline{t},  \overline{\lambda} u^{**} + (1-\overline{\lambda})u^{*} )\in S$. This is in contradiction to \eqref{twofactors7}.
		
		Next, assume that $\alpha \mu(\overline{u}) < -\alpha \beta$. Observe that there exists $t^{*}\in (\overline{t}, 1]$ such that $m\beta(t^{*} - \overline{t}) < -\alpha \beta- \alpha \mu (\overline{u}) $. Hence,
		\[
		\beta \theta(t^{*}) + \alpha \mu (\overline{u}) < \alpha \beta + m\beta ( t^{*} - \overline{t}) + \alpha \mu(\overline{u}) < 0,
		\]
		where the first inequality follows from (\ref{twofactors3}). Therefore $(t^{*},\overline{u}) \in S$.  Next, we can choose $(t^{**},u^{**})\in [0,\overline{t}] \times (\overline{u}, 1]$ such that 
		\[
		m\beta (t^{**} - \overline{t}) + n\alpha(u^{**}- \overline{u}) \leq 0.
		\]
		Therefore, $(t^{**},u^{**})\in T \cap [0,\overline{t}] \times (\overline{u}, 1]$. By (\ref{twofactors8}),  $(t^{**},u^{**})\in S$. Since $  t^{**} < \overline{t} < t^{*} $ and $\overline{u} < u^{**}$, there exists $\overline{\lambda}\in(0,1)$ such that
		\[
		\of{ \overline{\lambda}t^{**} + (1-\overline{\lambda})t^{*},  \overline{\lambda} u^{**} + (1-\overline{\lambda})\overline{u} } \in (\overline{t},1]\times(\overline{u},1].
		\]
		Moreover $S$ is convex. Hence, $( \overline{\lambda}t^{**} + (1-\overline{\lambda})\overline{t},  \overline{\lambda} u^{**} + (1-\overline{\lambda})u^{*}) \in S$. This is a contradiction to (\ref{twofactors7}).

		\item Suppose that $\overline{t} = 0$ and $\overline{u} > 0$. By (\ref{twofactors1}),  we have $\alpha - \theta(0) > 0 $. Choose $u^{*} \in (\overline{u},1]$ such that
		\[
			n\alpha (u^{*} - \overline{u}) < \beta(\alpha  - \theta(0)).
		\]
		For instance, one can choose $u^{*} = \overline{u} + \frac{1}{2} \frac{\beta( \alpha - \theta(0) )}{n \alpha}$. Observe that 
		\[
			\beta \theta(0) + \alpha \mu(u^{*})< \beta \theta(0) -\beta \alpha + n\alpha (u^{*} - \overline{u}) < 0,
		\]
		where the first strict inequality follows from (\ref{twofactors3}). By the definition of $S$, $(0, u^{*})\in S$.  Choose $(t^{**}, u^{**})\in(\overline{t},1]\times [0,\overline{u}]$ such that
		\[
			m\beta (t^{**} - \overline{t}) + n\alpha (u^{**} -\overline{u}) \leq 0.
		\]
		In other words, $(t^{**}, u^{**})\in T$. By (\ref{twofactors9}),  $(t^{**}, u^{**})\in S$. Since $S$ is convex and $(0, u^{*}),(t^{**}, u^{**})\in S$, for any $\lambda \in [0,1]$
		\[
			(\lambda t^{**} + (1-\lambda)0, \lambda u^{**} + (1-\lambda) u^{*})\in S.
		\]
		On the other hand, $ 0 = \overline{t} <  t^{**}$ and $ u^{**} \leq \overline{u} <  u^{*}$. Therefore, there exists $\overline{\lambda}\in (0,1)$ such that
		\[
			(\overline{\lambda} t^{**} + (1-\overline{\lambda})0, \overline{\lambda} u^{**} + (1-\overline{\lambda}) u^{*})\in (\overline{t}, 1] \times (\overline{u}, 1].
		\]
		This is a contradiction to (\ref{twofactors7}).

		\item Suppose that $\overline{t} > 0$ and $\overline{u} = 0$. By (\ref{twofactors1}),  we have $\mu(0) + \beta < 0$. Choose $t^{*} \in (\overline{t},1]$ such that
		\[
			m\beta (t^{*} - \overline{t}) < - \alpha (\beta  + \mu(0)). 
		\]
		For instance, one can choose $t^{*} = \overline{t} - \frac{1}{2} \frac{\alpha ( \beta - \mu(0) )}{m \beta}$. Observe that 
		\[
			\beta \theta(t^{*}) + \alpha \mu(0)< \alpha \beta +m\beta (t^{*} - \overline{t}) + \alpha\mu(0) <  0.
		\]
		where the first strict inequality follows from (\ref{twofactors3}). By the definition of $S$, $(t^{*},0)\in S$. Choose $(t^{**}, u^{**})\in[0,\overline{t}]\times(\overline{u}, 1]$ such that
		$m\beta (t^{**} - \overline{t}) + n\alpha (u^{**} -\overline{u}) \leq 0$.
		In other words, $(t^{**}, u^{**})\in T$. By (\ref{twofactors8}),  $(t^{**}, u^{**})\in S$. Since $S$ is convex and $(t^{*}, 0),(t^{**}, u^{**})\in S$, for any $\lambda \in [0,1]$,
		\[
			(\lambda t^{**} + (1-\lambda)t^{*}, \lambda u^{**} + (1-\lambda) 0)\in S.
		\]
		On the other hand, $ t^{**} \leq \overline{t} <  t^{*}$ and $ 0  = \overline{u} <  u^{**}$. Therefore, there exists $\overline{\lambda}\in (0,1)$ such that
		\[
			(\overline{\lambda} t^{**} + (1-\overline{\lambda})t^{*}, \overline{\lambda} u^{**} + (1-\overline{\lambda}) 0)\in (\overline{t}, 1] \times (\overline{u}, 1].
		\]
		This is a contradiction by (\ref{twofactors7}).
		\item Suppose that $\overline{t} = \overline{u} = 0$. By (\ref{twofactors1}), we have $\mu(0) + \beta < 0$ and $\alpha - \theta(0) > 0$. Choose $t^{*},u^{*} \in (0,1]$ such that 
		\[
			\beta m t^{*} < -\alpha(\beta  +  \mu(0)), \quad
			\alpha n u^{*} <  \beta(\alpha - \theta(0)).
		\]
		For instance, choose $t^{*} = \frac{1}{2} \frac{\beta(\alpha - \mu(0))}{\alpha n}$ and $u^{*} = - \frac{1}{2}\frac{\alpha (\beta + \mu(0) )}{\beta m}$. Observe that, 
		\begin{align*}
			&\beta \theta(t^{*}) + \alpha \mu(0) < \alpha \beta + \beta m (t^{*} - \overline{t}) + \alpha \mu (0) < 0, \\
			&\beta \theta(0) + \alpha \mu(u^{*}) < -\beta \alpha + \alpha n (u - \overline{u})  + \beta \theta(0) < 0,
		\end{align*}
		where the first strict inequalities on the two lines follow from \eqref{twofactors3}. By the definition of $S$, $(t^{*},0)\in S$ and $(0,u^{*})\in S$. Since $S$ is convex, $(\lambda t^{*}, (1- \lambda) u^{*})\in S$  for any $\lambda\in (0,1)$. On the other hand, $(\lambda t^{*}, (1- \lambda) u^{*})\in (\overline{t}, 1] \times (\overline{u}, 1]$, which is a contradiction by (\ref{twofactors7}).
	\end{enumerate}
	Since (a), (b), (c) and (d) all lead to a contradiction, we conclude that $c(f) + c(g) \geq 0$.
	
	
	Conversely, assume that $c(f) + c(g) \geq 0$. Since $f$ and $g$ are non-constant, by Theorem \ref{convindex_values_infinitedimensional_constant}, $c(f)$ and $c(g)$ are finite. If $c(f)\geq 0$ and $c(g)\geq 0$, by Theorem \ref{convindex_values_infinitedimensional}, $f,g$ are convex. As the sum of two convex functions, $s$ is convex, hence also quasiconvex. If $c(f)$ and $c(g)$ are not both positive, then by symmetry it is enough to consider the case $c(f) < 0 < c(g)$. Define,
	\[
		\psi(x) \coloneqq e^{c(g)f(x)},\ x\in \X;\quad \quad 
		\zeta(y) \coloneqq e^{-c(g)g(y)},\ y\in \Y.
	\]
	It follows from Definition \ref{convindex_def_infdimensional} and  Lemma \ref{lemma 1 section 2}-(i) that $x \mapsto e^{-c( f )f(x)}$ is convex. Since $-c(g) \leq c(f)$, by Lemma \ref{lemma 1 section 2}-(i), $\psi$ is convex. Moreover, $\zeta$ is concave directly from Definition \ref{convindex_def_infdimensional} and Lemma \ref{lemma 1 section 2}-(ii). On $\X\times \Y$, let $\rho$ be defined by
	\[
	\rho(x, y) \coloneqq e^{c(g)[f(x) + g(y)]} = \frac{\psi(x)}{\zeta(y)}.
	\]
	Notice that for every $(x_{1},y_{1}),(x_{2},y_{2})\in \X\times \Y$ and $\lambda \in [0,1]$,
	\begin{equation}  \label{eq:theoremsum1}
		\rho\big(\lambda (x_{1},y_{1}) + (1-\lambda) (x_{2},y_{2}) \big)
		= \frac{\psi(\lambda x_{1} + (1-\lambda)x_{2})}{\zeta(\lambda y_{1} + (1-\lambda) y_{2})}\\
		\leq \frac{\lambda\psi(x_{1}) + (1-\lambda)\psi(x_{2})}{\lambda \zeta(y_{1}) + (1-\lambda)\zeta(y_{2})},
	\end{equation}
	where the inequality follows from the convexity of $\psi$ and the concavity of $\zeta$. 
	
	First, suppose that $\frac{\psi(x_{2})}{\zeta(y_{2})} \leq \frac{\psi(x_{1})}{\zeta(y_{1})}$. Then,
	\[
	\zeta(y_{1}) [\lambda \psi(x_{1}) +  (1-\lambda)\psi(x_{2})]   \leq   \psi(x_{1})[\lambda \zeta(y_{1})  + (1-\lambda)\zeta(y_{2})],
	\]
	which implies
	\[
		\frac{  \lambda \psi(x_{1}) +  (1-\lambda)\psi(x_{2})  }{ \lambda \zeta(y_{1})  + (1-\lambda)\zeta(y_{2})   } \leq \frac{  \psi(x_{1})  }{\zeta(y_{1}) }.
	\]
	Therefore,
	\begin{equation}\label{eq:theoremsum2}
		\frac{\lambda \psi(x_{1}) + (1-\lambda)\psi(x_{2})}{\lambda \zeta(y_{1}) + (1-\lambda)\zeta(y_{2})} \leq \rho(x_{1}, y_{1}).
	\end{equation}		
	Next, suppose that $\frac{\psi(x_{2})}{\zeta(y_{2})}  >  \frac{\psi(x_{1})}{\zeta(y_{1})}$. A similar argument as in the previous case gives
	\begin{equation}\label{eq:theoremsum3}
		\frac{  \lambda \psi(x_{1}) +  (1-\lambda)\psi(x_{2})  }{ \lambda \zeta(y_{1})  + (1-\lambda)\zeta(y_{2})   } < \rho(x_{2}, y_{2}).
	\end{equation}
	By \eqref{eq:theoremsum1}, \eqref{eq:theoremsum2} and \eqref{eq:theoremsum3},
	$	\rho(\lambda (x_{1},y_{1}) + (1-\lambda) (x_{2},y_{2}) ) \leq \max\{\rho(x_{1}, y_{1}), \rho(x_{2}, y_{2}) \}$.
	Therefore, $\rho$ is quasiconvex. Furthermore,
	$	(x,y)\mapsto \log\of{ \rho(x,y) }= c(g)[f(x) + g(y)] = c(g)s(x, y)$
	is quasiconvex since quasiconvexity is stable under composition with a nondecreasing function. Hence, $s$ is quasiconvex.
\end{proof}

The following lemma is needed in the proof of Proposition \ref{ConvIndex_AdditiveLem_2}. As in the proof of Theorem \ref{TwoFactorCase}, we write $\X_{1} = \X$ and $\X_{2} = \Y$.

\begin{lem} \label{ConvIndex_AdditiveLem_1}
	Let $\rho_{1},\rho_{2}$ be positive real-valued lower semicontinuous and convex functions on $\X, \Y$, respectively. On $\X\times \Y$, define $\rho$ by
	\[
	\rho(x, y) \coloneqq [\rho_{1}(x)]^{\alpha}[\rho_{2}(y)]^{\beta},
	\]
	where $\alpha,\beta > 0$ and $\alpha + \beta = 1$. If $\rho$ is concave, then at least one of $\rho_{1}$ and $\rho_{2}$ is concave. 
\end{lem}
\begin{proof}
	Assume that $\rho$ is concave and suppose to the contrary that both $\rho_{1}$ and $\rho_{2}$ are not concave. Choose $\overline{y}\in \Y$ such that $\rho_{2}(\overline{y}) > 0$. By the concavity of $\rho$, for every $x_{1}, x_{2} \in \X$ and $\lambda\in [0,1]$,
	\[
		\rho(\lambda x_{1} + (1-\lambda)x_{2}, \overline{y}) \geq \lambda \rho(x_{1}, \overline{y}) + (1-\lambda)\rho(x_{2}, \overline{y}).
	\]
	Dividing both sides by $[\rho_{2}(\overline{y})]^{\beta}$ yields
	\begin{equation} \label{ConvIndex_AdditiveLem_1-1}
		[\rho_{1}(\lambda x + (1-\lambda) x_{2})]^{\alpha} \geq \lambda [\rho_{1}(x_{1})]^{\alpha} + (1-\lambda) [\rho_{1}(x_{2})]^{\alpha}.
	\end{equation}
	Note that $\lambda [\rho_{1}(x_{1})]^{\alpha} + (1-\lambda) [\rho_{1}(x_{2})]^{\alpha} \geq \min\{ [\rho_{1}(x_{1})]^{\alpha}, [\rho_{1}(x_{2})]^{\alpha} \} =  \min\{ \rho_{1}(x_{1}), \rho_{1}(x_{2}) \} ^{\alpha}$. So
	\[
		[\rho_{1}(\lambda x_{1} + (1-\lambda)x_{2})]^{\alpha} \geq \min\{ \rho_{1}(x_{1}), \rho_{1}(x_{2}) \} ^{\alpha},
	\]
	which implies
	\begin{equation}\label{ConvIndex_AdditiveLem_1-2}
		\rho_{1}(\lambda x_{1} + (1-\lambda)x_{2}) \geq \min\{ \rho_{1}(x_{1}), \rho_{1}(x_{2}) \}.
	\end{equation}
	The supposition together with \eqref{ConvIndex_AdditiveLem_1-2} implies that $-\rho_{1}$ is quasiconvex but not convex. By Lemma \ref{Lemma 2.2 - 1}, there exist $x_{1}, x_{2}\in \X$, $\overline{t}\in[0,1)$ and $\gamma \in [0,1]$ such that
	\begin{align}
		&\theta(0) > \gamma \geq \theta(1) > 0, \label{ConvIndex_AdditiveLem_1-3} \\
		&\theta(t) \geq \gamma - m (t - \overline{t}),\quad t\in[0,\overline{t}], \label{ConvIndex_AdditiveLem_1-4}\\
		&\gamma \geq \theta(t) >  \gamma - m (t - \overline{t}),\quad t\in (\overline{t},1], \label{ConvIndex_AdditiveLem_1-5}
	\end{align}
	where $\theta(t) \coloneqq \rho_{1}(t x_{1} + (1-t) x_{2})$, $t\in[0,1]$, and $m \coloneqq \theta(0) - \theta(1) > 0$. 
	
	The function $t\mapsto \overline{\theta} (t) \coloneqq  [\theta(t)]^{\alpha}$ is concave on $[0,1]$ by \eqref{ConvIndex_AdditiveLem_1-1} and the fact that the composition of a concave function with an affine function is concave. As a real-valued concave function defined on [0,1], $\overline{\theta}$  is continuous on $(0,1)$ and have finite right and left derivatives. Consequently, $\theta = \overline{\theta}^{\frac{1}{\alpha}}$ is continuous as it is the composition of two continuous.  It is not hard to see that
	\begin{equation} \label{ConvIndex_AdditiveLem_1-Chain Rule}
			\overline{\theta}^{\prime}_{+}(t) = \alpha ( \theta(t) ) ^{\alpha - 1} \theta_{+}^{\prime}(t), \qquad
			\overline{\theta}^{\prime}_{-}(t) = \alpha ( \theta(t) ) ^{\alpha - 1} \theta_{-}^{\prime}(t),
	\end{equation}
	respectively. Hence, $\theta$ also has finite right and left derivatives. Moreover,by \eqref{ConvIndex_AdditiveLem_1-3} we have
	\begin{equation}\label{ConvIndex_AdditiveLem_1-inter result}
		\overline{\theta}(0) > \gamma^{\alpha}.
	\end{equation}
	We claim that $\overline{t} \neq 0$. Indeed, supposing otherwise, by \eqref{ConvIndex_AdditiveLem_1-5}, we get $\gamma \geq \theta(t)$ for every $t\in (0,1]$. Then, since $\overline{\theta}$ is right-continuous at $0$, we have $\gamma^{\alpha} \geq \overline{\theta}(0)$. This is in contradiction with \eqref{ConvIndex_AdditiveLem_1-inter result}; therefore, $\overline{t} \neq 0,$ and consequently $\overline{\theta}$ is continuous at $\overline{t}$. It follows from \eqref{ConvIndex_AdditiveLem_1-4} and \eqref{ConvIndex_AdditiveLem_1-5} that $\theta(\overline{t}) =  \gamma$ and $\theta^{\prime}_{-}(\overline{t}) \leq -m \leq \theta^{\prime}_{+}(\overline{t})$. Furthermore, $\overline{\theta}^{\prime}_{+}(\overline{t}) \leq \overline{\theta}^{\prime}_{-}(\overline{t})$ since $\overline{\theta}$ is concave. In view of \eqref{ConvIndex_AdditiveLem_1-Chain Rule}, $\theta^{\prime}_{+}(\overline{t}) = \theta^{\prime}_{-}(\overline{t}) = -m$. Thus, $\theta$ is differentiable at $\overline{t}$ and $\theta^{\prime}(\overline{t}) < 0$. By \eqref{ConvIndex_AdditiveLem_1-5}, we get
	\begin{equation}\label{ConvIndex_AdditiveLem_1-6}
		\theta(t) > \theta(\overline{t}) + \theta^{\prime}(\overline{t})(t-\overline{t}),\quad t\in (\overline{t}, 1].
	\end{equation}
	
	Reasoning as above, we can prove that there exist $y_{1}, y_{2}\in \Y$, $\overline{u}\in (0,1)$ such that $\mu^{\prime}(\overline{u}) < 0$ and
	\begin{equation}\label{ConvIndex_AdditiveLem_1-7}
		\mu(u) > \mu(\overline{u}) + \mu^{\prime}(\overline{u}) (u-\overline{u}),\quad u\in (\overline{u}, 1],
	\end{equation}
	where
	$	\mu(u) \coloneqq \rho_{2}(u y_{1} + (1-u) y_{2})$, $u\in[0,1]$.
	
	Let
	\[
		\zeta(t,u) \coloneqq [\theta(t)]^{\alpha}[\mu(u)]^{\beta},\quad t,u\in[0,1];
	\]
	and, for every $\delta > 0$, define
	\[
		h_{\delta} \coloneqq -\delta \frac{\theta(\overline{t})}{\theta^{\prime}(\overline{t})},\qquad
		k_{\delta} \coloneqq -\delta \frac{\mu( \overline{u} ) }{\mu^{\prime}(\overline{u})}.
	\]
	Next, choose $\delta$ so that $0 < h_{\delta} < 1-\overline{t}$ and $0 < k_{\delta} < 1-\overline{u}$. For instance, one can choose $\delta = \frac{1}{2} \min \{ -\frac{\theta^{\prime}(\overline{t})}{\theta(\overline{t}) } (1-\overline{t}) ,  -\frac{\mu^{\prime}(\overline{u})}{\mu(\overline{u}) } (1-\overline{u})  \}$. By \eqref{ConvIndex_AdditiveLem_1-5} and \eqref{ConvIndex_AdditiveLem_1-7}, we have
	\begin{align*}
		\theta(\overline{t} + h_{\delta}) &> \theta(\overline{t}) + \theta^{\prime}(\overline{t})h_{\delta} = \theta(\overline{t})[1-\delta]\\
		\mu(\overline{u}+ k_{\delta}) &> \mu(\overline{u}) + \mu^{\prime}(\overline{u}) k_{\delta} = \mu(\overline{u})[1-\delta]
	\end{align*}
	Hence,
	\begin{equation}\label{ConvIndex_AdditiveLem_1-8}
		\zeta(\overline{t} + h_{\delta},\overline{u} + k_{\delta}) > [\theta(\overline{t})]^{\alpha}[\mu(\overline{u})]^{\beta}[1-\delta].
	\end{equation}
	On the other hand, for every $\lambda\in[0,1]$,
	\begin{align*}
		\zeta(\overline{t} + \lambda h_{\delta},\overline{u} + \lambda k_{\delta})
		= \zeta\big( (\overline{t},\overline{u}) + \lambda (h_{\delta},k_{\delta}) \big)
		&= \zeta \big(\lambda (\overline{t} + h_{\delta}, \overline{u} + k_{\delta}) + (1-\lambda) (\overline{t}, \overline{u})\big) \\
		&\geq \lambda \zeta (\overline{t} + h_{\delta}, \overline{u} + k_{\delta} ) + (1- \lambda ) \zeta (\overline{t}, \overline{u}) \\
		&= \zeta(\overline{t},\overline{u}) + \lambda  [ \zeta(\overline{t}+ h_{\delta},\overline{u} + k_{\delta}) -  \zeta(\overline{t},\overline{u}) ],
	\end{align*}
	where the inequality follows from concavity of $\zeta$ on $[0,1] \times [0,1]$. With some arithmetic, we get
	\[
		\zeta(\overline{t},\overline{u}) + \frac{\zeta\big( (\overline{t},\overline{u}) + \lambda(h_{\delta},k_{\delta})\big) -  \zeta(\overline{t},\overline{u})}{\lambda} \geq \zeta(\overline{t}+ h_{\delta},\overline{u} + k_{\delta}) 
	\]
	Moreover, the directional derivative of $\zeta$ along the direction $(h_{\delta},k_{\delta})$ and $(\overline{t}, \overline{u})$ exists. By letting $\lambda\to 0$, we get
	\[
		\zeta(\overline{t},\overline{u}) + \frac{\partial \zeta}{\partial t}(\overline{t}, \overline{u}) h_{\delta} + \frac{\partial \zeta}{\partial u}(\overline{t}, \overline{u}) k_{\delta} \geq \zeta(\overline{t}+ h_{\delta},\overline{u} + k_{\delta}).
	\]
	Plugging in the values of the partial derivatives yields
	\[
		\zeta(\overline{t},\overline{u})\Big[1 + \alpha h_{\delta} \frac{\theta^{\prime}(\overline{t})}{\theta(\overline{t})} + \beta k_{\delta} \frac{\mu^{\prime}(\overline{u})}{\mu(\overline{u})}\Big] \geq \zeta(\overline{t}+ h_{\delta},\overline{u} + k_{\delta}),
	\]
	It follows from the defintion of $h_{\delta}$, $k_{\delta}$ and $\zeta$ that 
	\[
		[\theta(\overline{t})]^{\alpha}[\mu(\overline{u})]^{\beta}[1-\delta] \geq \zeta(\overline{t} + h_{\delta}, \overline{u} + k_{\delta}),
	\]
	which is in contradiction with \eqref{ConvIndex_AdditiveLem_1-8}. Therefore, at least one of the functions $\rho_{1}$ and $\rho_{2}$ is concave. 
\end{proof}

The  next proposition gives an additive formula for the convexity index. 
\begin{prop} \label{ConvIndex_AdditiveLem_2}
	Assume that $f_{i}$ is lower semicontinuous and convex for  each $i \in  \{ 1,\ldots n \}$. Then,
	\begin{equation} \label{eq:ConvIndexSumFormula}
		\frac{1}{c(s)} = \sum_{i = 1}^{n} \frac{1}{c(f_{i})},
	\end{equation}
	where $s$ is defined by \eqref{Section 2 - Sum Function}.
\end{prop}
\begin{proof}
	It is enough to consider the following three cases with $n= 2$. Note that, since $s$ is convex, by Theorem \ref{convindex_values_infinitedimensional} implies that $c(s) \geq 0$.
	\begin{enumerate}[(a)]
		\item Suppose that $c(f_{1}) = 0$ and $c(f_{2}) \in [0, + \infty)$. By Definition \ref{convindex_def_infdimensional}, $r_{\lambda}$ associated with $f_{1}$ is not concave for any $\lambda > 0$. Suppose for a contradiction that there exists $\overline{\lambda} > 0$ such that $r_{\overline{\lambda}}$ associated with $s$ is concave. Then, for every $(x_{1},x_{2}), (x_{1}^{\prime},x_{2})\in \X_{1}\times \dom f_{2}$, and $ t \in[0,1]$
		\[
			e^{-\overline{\lambda}[f_{1}(t x_{1} + (1-t) x_{1}^{\prime}) + f_{2}(x_{2})]} \geq t e^{-\overline{\lambda}[f_{1}(x_{1}) + f_{2}(x_{2})]} + (1-t) e^{-\overline{\lambda}[f_{1}(x_{1}^{\prime}) + f_{2}(x_{2})] }.
		\]
		Dividing both sides by $e^{-\overline{\lambda}f_{2}(x_{2})}$ yields
		\[
			e^{-\overline{\lambda}f_{1}(\eta x_{1} + (1-\eta) x_{1}^{\prime})} \geq \eta e^{-\overline{\lambda}f_{1}(x_{1})} + (1-\eta) e^{-\overline{\lambda}f_{1}(x_{1}^{\prime}) },
		\]
		where $\eta \coloneqq \frac{t}{e^{-\overline{\lambda}f_{2}(x_{2})}}$. In other words, $r_{\overline{\lambda}}$ associated with $f_{1}$ is concave. This  is in contradiction with the assumption $c(f_{1}) = 0$. Therefore, $r_{\lambda}$ associated with $s$ is not concave for every $\lambda > 0$. By Definition  \ref{convindex_def_infdimensional}, $c(s) = 0$ and \eqref{eq:ConvIndexSumFormula} holds.
		\item Suppose that $c(f_{1}) = +\infty$ and $c(f_{2}) \in [0, + \infty]$. By Theorem \ref{convindex_values_infinitedimensional_constant}, $f_{1} \equiv K$, where $K\in \R \cup \{ + \infty \}$ is a constant. If $K = +\infty$, then $s(x_{1}, x_{2}) = +\infty$ for every $(x_{1}, x_{2}) \in \X_{1}\times \X_{2}$. Therefore, \eqref{eq:ConvIndexSumFormula} holds trivially. If $c(s) = + \infty$, then $f_{1}, f_{2}$ are constant. Hence, \eqref{eq:ConvIndexSumFormula} holds trivially. So assume $K > +\infty$ and $c(f_{2}) < +\infty$. Observe that
		\[
			e^{-c(f_{2}) s(x_{1}, x_{2})} = e^{-c(f_{2}) K } e^{-c(f_{2}) f_{2}(x_{2})},\quad (x_{1}, x_{2}) \in \X_{1}\times \X_{2}. 
		\]
		$e^{-c(f_{2}) K } > 0$ is a constant and $x_{2} \mapsto e^{-c(f_{2}) f_{2}(x_{2})}$ is concave on $\X_{2}$. Therefore,  $r_{c(f_{2})}$ associated with $s$ is concave. Definition \ref{convindex_def_infdimensional} together with Lemma \ref{lemma 1 section 2} implies that $c(f_{2}) \leq c(s)$.
		
		On the other hand, $e^{-c(s) s}$ is concave by Definition \ref{convindex_def_infdimensional}. Moreover,
		$e^{-c( s ) s} = e^{-c(s)K} e^{-c(s) f_{2}}$, where $e^{-c(s)K} > 0$ is a constant. Therefore,  $e^{-c(s)f_{2}}$ is concave.  It follows from Definition \ref{convindex_def_infdimensional} and Lemma \ref{lemma 1 section 2} that $c(s) \leq c(f_{2})$. Hence, $c(s) = c(f_{2})$ and \eqref{eq:ConvIndexSumFormula} holds.
		\item Suppose that $0 < c(f_{i}) < +\infty$ for $i \in\{1,2\}$. Set $\alpha_{i} = \frac{1}{c(f_{i})}$ for $i\{1,2\}$ and $\alpha = \alpha_{1} + \alpha_{2}$. Then,
		\[
			e^{-\frac{1}{\alpha} s(x_{1}, x_{2}) } = \Big(e^{- c(f_{1}) f_{1}(x_{1}) } \Big)^{ \frac{\alpha_{1}}{\alpha} } \Big(e^{- c(f_{2}) f_{2}(x_{2}) } \Big)^{ \frac{\alpha_{2}}{\alpha} }.
		\]
		Note that $x_{i} \mapsto e^{- c(f_{i}) f_{i}(x_{i}) }$ is concave on $\X_{i}$ for $i \in \{1,2\}$ and $g(u,v) = u^{\frac{\alpha_{1}}{\alpha}} v^{\frac{\alpha_{2}}{\alpha}}$ is non-decreasing in each argument, concave on $\R_{+}\times \R_{+}$. Therefore, $(x_{1}, x_{2}) \mapsto e^{-\frac{1}{\alpha} s(x_{1}, x_{2}) }$ is concave on $\X_{1}\times \X_{2}$ as it is the composition of a concave functions with a function which is non-decreasing in each argument. In view of Definition \ref{convindex_def_infdimensional}, we have
		\begin{equation}\label{eq:ConvIndex_AdditiveLem_2_1}
			\frac{1}{\alpha} \leq c(s) \iff \frac{1}{c(f_{1})} + \frac{1}{c(f_{2})} \geq \frac{1}{c(s)}.
		\end{equation}
		
		Next, assume that $\mu\in [0, \alpha)$. We have,
		\[
			e^{-\frac{1}{\mu} f_{i}(x_{i}) } = \Big(e^{-\frac{\alpha}{\mu} f_{i}(x_{i})}\Big)^{\frac{1}{ \alpha } } =  \Big(e^{-\frac{\alpha}{\mu} c(f_{i}) f_{i}(x_{i})}\Big)^{\frac{\alpha_{i}}{\alpha} }, \quad x_{i}\in \X_{i}
		\]
		for each $i \in \{ 1,2\}$. Since $\frac{\alpha}{\mu} > 1$, it follows from Definition \ref{convindex_def_infdimensional} that $x_{i} \mapsto e^{-\frac{\alpha}{\mu} c(f_{i}) f_{i}(x_i)}$ is not concave for each $i \in \{ 1,2 \}$. Moreover,
		\[
			e^{-\frac{1}{\mu} s(x_{1}, x_{2})} = \Big( e^{-\frac{\alpha}{\mu} c(f_{1}) f_{1}(x_{1})}  \Big)^{\frac{\alpha_{1}}{\alpha} }  \Big( e^{-\frac{\alpha}{\mu} c(f_{2}) f_{2}(x_{2})}  \Big)^{\frac{\alpha_{2}}{\alpha} }, \quad (x_{1}, x_{2}) \in \X_{1}\times \X_{2}.
		\]
		In view of Lemma \ref{ConvIndex_AdditiveLem_1}, $e^{-\frac{1}{\mu} s}$ is not concave and consequently $c(s) < \frac{1}{\mu}$. Letting $\mu \to \alpha$, we have $c(s) \leq \frac{1}{\alpha}$. Thus,
		\begin{equation}\label{eq:ConvIndex_AdditiveLem_2_2}
			\frac{1}{c(f_{1})} + \frac{1}{c(f_{2})} \leq \frac{1}{c(s)}.
		\end{equation}
		Summing up, \eqref{eq:ConvIndex_AdditiveLem_2_1} and \eqref{eq:ConvIndex_AdditiveLem_2_2} imply \eqref{eq:ConvIndexSumFormula}.
	\end{enumerate}
\end{proof}	

We are ready to give the second main result of this section.
\begin{thm} \label{MainResult-2.2}
	Assume that $f_{i}$ is a non-constant lower semicontinuous function for each $i \in \{1, \ldots, n \}$. Then, $s$ is quasiconvex if and only if one of the following conditions holds:
	\begin{itemize}
		\item[(i)] $f_{1}, \ldots, f_{n}$ are convex.
		\item[(ii)] All of $f_{1}, \ldots f_{n}$ except one are convex and
		\begin{equation}\label{eq:MainResult-2.2}
			\frac{1}{c(f_{1})}  + \ldots + \frac{1}{c(f_{n})} \leq 0.
		\end{equation}
	\end{itemize}
\end{thm}
\begin{proof}
	Assume that $s$ is quasiconvex. Then, $f_{1}, \ldots ,f_{n}$ are quasiconvex. We claim that at most one of the functions $f_{1}, \ldots f_{n}$ is not convex. Suppose for a contradiction that $f_{i_{1}}, \ldots, f_{i_{k}}$ are not convex for some $k \geq 2$ so that from Theorem \ref{convindex_values_infinitedimensional} we deduce $c(f_{i_{j}}) < 0$ for each $j\in \{1,\ldots k\}$. Thus
	\begin{equation}\label{eq:MainResult-2.2-1}
		c(f_{i_{1}}) + \ldots + c(f_{i_{k}}) < 0.
	\end{equation}
	Next, define $\overline{s}$ on $ \X_{i_{1}} \times \ldots \times \X_{i_{k}}$ by 
	\[
		\overline{s}(x_{i_{1}}, \ldots, x_{i_{k}}) \coloneqq f_{i_{1}}(x_{i_{1}}) + \ldots + f_{i_{k}}(x_{i_{k}}).
	\]
	Clearly, $\overline{s}$ is quasiconvex. Then, by Theorem \ref{TwoFactorCase},
	$c(f_{i_{1}}) + \ldots + c(f_{i_{k}}) \geq 0$, which is in contradiction with \eqref{eq:MainResult-2.2-1}. Hence, the claim holds. If $f_{1}, \ldots, f_{n}$ are convex, then we are done. Suppose that $f_{m}$ is not convex for some $m\in \{1, \ldots, n\}$ and define $\hat{s}$ on $\bigtimes_{i \in \{1, \ldots, n\} \setminus \{m\}} \X_{i}$ by
	\[
		\hat{s}(x_{1}, \ldots, x_{m-1}, x_{m+1}, \ldots, x_{n}) = \sum_{i \in \{1, \ldots, n\} \setminus \{m\}} f_{i} (x_{i}).
	\]
	By Theorem \ref{TwoFactorCase}, $c(\hat{s}) + c(f_{m}) \geq 0$.
	Since $c(f_{m}) < 0$, we have $\frac{1}{c(\hat{s})} \leq - \frac{1}{c(f_{m})}$, which implies by Proposition \ref{ConvIndex_AdditiveLem_2} that
	\[
		\frac{1}{c(f_{1})} + \ldots + \frac{1}{c(f_{n})} \leq 0.
	\]
	
	Conversely, assume that either $f_{1}, \ldots, f_{n}$ are convex or all of $f_{1}, \ldots, f_{n}$  except one are convex and \eqref{eq:MainResult-2.2} holds. If the former holds, then $s$ is quasiconvex trivially. Suppose that the latter holds together with $f_{m}$ is not convex for some $m\in \{1, \ldots, n\}$ and $\hat{s}$ is defined as above. Proposition \ref{ConvIndex_AdditiveLem_2} together with \eqref{eq:MainResult-2.2} implies that
	\[
		\frac{1}{c(\hat{s})} + \frac{1}{c(f_{m})}\leq 0.
	\]
	Since $c(f_{m}) < 0$, this implies that $c(\hat{s}) + c(f_{m}) \geq 0$. In view of Theorem \ref{TwoFactorCase}, $s$ is quasiconvex.
\end{proof}


\subsection{Infinite decomposable sums}\label{subsec:infsum}

The aim of this section is to extend Theorem \ref{MainResult-2.2} to infinite decomposable sums. Similar to the notation of the previous sections, $\X_{i}$ is a topological vector space and $f_{i}$ is a proper extended real-valued function on $\X_{i}$ for each $i \in \N$. Let $\X = \bigtimes_{i = 1}^{\infty} \X_{i}$. For every $n\in\N$ and $x \in \X$, let
\[
	s_{n}(x) \coloneqq \sum_{i = 1}^{n} f_{i}(x_{i}),\qquad s(x) \coloneqq \lim_{n\to\infty} s_{n}(x)
\]
provided that the limit exists. Given $n\in\N$, it can be easily noticed that if $f_i$ is lower semicontinuous on $\X_i$ for each $i \in\{ 1,\ldots,n\}$, then $s_{n}$ is lower continuous on $\X$ since, in this case, each $f_{i}$ is also lower semicontinuous on $\X$ and lower semicontinuity is preserved under finite sums. Before presenting the main result of this section, we give two separate conditions under which lower semicontinuity is preserved under infinite sums. 
\begin{lem} \label{Lemma Section 2.3}
	Assume that $f_{i}$ is a lower semicontinuous function for each $i \in\N$ and at least one of the following conditions holds:
	\begin{enumerate}[(i)]
		\item $f_{i}$ is real-valued for each $i \in \N$ and $(s_{n})_{n\in\N}$ is uniformly convergent.
		\item $f_{i} \geq 0$ for each $i \in \N$. 
	\end{enumerate}
	Then, $s$ is well-defined and lower semicontinuous.
\end{lem}
\begin{proof}
	In both cases, it is obvious that $s$ is well-defined. Assume that $(i)$ holds. Let $x_{0}\in \X$. Note that $s_{n}(x_{0})\in\R$ since each $f_{i}$ is real-valued. By (uniform) convergence, $s(x_0)\in\R$ as well. Thus,
	\begin{align}\label{eq:Lemma Section 2.3-1}
			s(x_{0}) 
			= s(x_{0}) - s_{n} (x_{0}) + s_{n} (x_{0}) 
			 \leq s_{n}(x_{0}) + |s(x_{0}) - s_{n}(x_{0})| 
			 \leq  s_{n}(x_{0}) + \sup_{x\in \X}|s(x) - s_{n}(x)|.
	\end{align}
	It follows from the lower semicontinuity of $s_{n}$ at $x_{0}$ together with \eqref{eq:Lemma Section 2.3-1} that
	\begin{equation}\label{eq:Lemma Section 2.3-2}
		s(x_{0}) \leq \sup_{x\in \X}|s(x) - s_{n}(x)| + \liminf_{x\to x_{0} }s_{n}(x).
	\end{equation}
	Moreover, similar to \eqref{eq:Lemma Section 2.3-1}, for every $x^{\prime}\in \X$ and $n\in\N$,
	\[
		s_{n}(x^{\prime}) \leq \sup_{x\in \X}|s(x) - s_{n}(x)| + s(x^{\prime}).
	\]
	By taking limit infima of both sides as $ x^{\prime} \to x_{0}$, we get
	\[
		\liminf_{x^{\prime} \to x_{0} } s_{n}(x^{\prime}) \leq \sup_{x\in \X}|s(x) - s_{n}(x)| + \liminf_{x^{\prime} \to x_{0} } s(x^{\prime}).
	\]
	Together with \eqref{eq:Lemma Section 2.3-2}, we have
	\[
		s(x_{0}) \leq 2 \sup_{x\in \X}|s(x) - s_{n}(x)| + \liminf_{x \to x_{0} } s(x).
	\]
	It follows from uniform convergence of $(s_{n})_{n\in\N}$ that
	\[
		s(x_{0}) \leq \liminf_{x \to x_{0} } s(x).
	\]
	
	Next, assume that $(ii)$ holds. For every $r\geq 0$ and $n\in\N$, define
	\[
		U_{r} \coloneqq \{ x\in \X \colon s(x) > r \},\quad 
		U_{r}^{n} \coloneqq \{ x\in \X \colon s_{n}(x) > r \}.
	\]
	Fix $r\geq 0$. We claim that
	\begin{equation} \label{eq:Lemma Section 2.3-3}
		U_{r} = \bigcup_{n= 1}^{\infty}U_{r}^{n}.
	\end{equation}
	Let $\overline{x} \in U_{r}$. Observe that $(s_{n}(\overline{x}))_{n\in\N}$ is a nondecreasing sequence in $\R$ since $f_{i}(\overline{x})\geq 0$ for each $i\in\N$. Hence, there exists $N \in \N$ such that $s_{N}(\overline{x}) > r$. Then, $\overline{x}\in U_{r}^{N}$, and consequently, $U_{r} \subseteq \bigcup_{n= 1}^{\infty}U_{r}^{n}$. Conversely, let $x^{\prime}\in \bigcup_{n= 1}^{\infty} U_{r}^{n}$. Then, there exists $N^{\prime}\in \N$ such that $x^{\prime} \in U_{r}^{N^{\prime}}\subseteq U_{r}$ since $(s_{n}(x^{\prime}))_{n\in\N}$ is a nondecreasing sequence in $\R$. Hence, $\bigcup_{n= 1}^{\infty}U_{r}^{n} \subseteq U_{r}$, and consequently, \eqref{eq:Lemma Section 2.3-3} holds. Finally, for each $n\in\N$, the set $U_{r}^{n}$ is open since $s_{n}$ is lower semicontinuous on $\X$. By \eqref{eq:Lemma Section 2.3-3}, $U_r$ is also open. Hence, $s$ is lower semicontinuous on $\X$.
\end{proof}
We are ready to give the main result of this section. 

\begin{thm} \label{MainResult-2.3}
	For each $i\in\N$, suppose that $f_{i}$ is a non-constant lower semicontinuous function. In addition, assume that either $f_{i} \geq 0$ for each $i\in\N$, or $f_{i}$ is real-valued for each $i \in\N$ as well as $(s_{n})_{n\in\N}$ is uniformly absolutely convergent. Then, $s$ is quasiconvex if and only if one of the following conditions holds:
	\begin{itemize}
		\item[(i)] $f_{i}$ is convex for each $i\in\N$.
		\item[(ii)] All of $f_{1},f_{2}, \ldots$ except one are convex and
		\begin{equation}\label{eq:MainResult-2.3-1}
			\sum_{i = 1}^{\infty} \frac{1}{c(f_{i})} \leq 0.
		\end{equation}
	\end{itemize}
\end{thm}
\begin{proof}
	Assume that $s$ is quasiconvex. Then, $f_{i}$ is quasiconvex for each $i\in\N$. We claim that $f_{i}$ is not convex for at most one $i\in \N$. Suppose not. Let us fix $n_{0}\in \N$ such that $n_{0} \geq 2$ and define $\overline{s}$ on $\bigtimes_{i = n_{0}+1}^{\infty} \X_{i}$ by
	\[
		\overline{s}(x_{n_{0} +1}, x_{n_{0} +2}, \ldots ) \coloneqq \sum_{j = n_{0}+1}^{\infty} f_{j}(x_{j}).
	\]
	Then, it follows from Theorem \ref{MainResult-2.2} that at most one of $f_{1}, \ldots f_{n_{0} }, \overline{s}$ is not convex. We need to consider the following two cases: 
	\begin{enumerate}[(a)]
		\item Suppose that $\overline{s}$ is convex. Then, at most one of $f_{1}, \ldots, f_{n_{0}}$ is not convex.
		\item Suppose that $\overline{s}$ is not convex. Let us choose $i_{0}, i_{1} > n_{0}$ such that $f_{i_{0}}, f_{i_{1}}$ are not convex. It follows from the quasiconvexity of $s$ on $\X$ that the function $\hat{s}$ defined on $\X_{1}\times\ldots \times \X_{n_{0}}\times \X_{i_{0}}\times \X_{i_{1}}$ by
		\[
			\hat{s} (x_{1},\ldots x_{n_{0}},x_{i_{0}}, x_{i_{1}}) \coloneqq f_{1}(x_{1}) + \ldots + f_{n_{0}}(x_{n_{0}}) + f_{i_{0}}(x_{i_{0}}) + f_{i_{1}}(x_{i_{1}})
		\]
		is quasiconvex. But then, by Theorem \ref{MainResult-2.2}, at least one of $f_{i_{0}}, f_{i_{1}}$ needs to be convex.
	\end{enumerate}
	Therefore, the claim holds, that is, either $f_{i}$ is convex for each $i\in\N$ or all of $f_{1},f_{2}, \ldots$ except one are convex. In the former case, $(i)$ holds. Suppose that the latter holds. Without loss of generality, since $(s_{n})_{n\in\N}$ is absolutely convergent, we can assume that $f_{1}$ is the function which is quasiconvex but not convex. For every $n\in \N$, define
	\[
		a_{n} \coloneqq   \sum_{i = 1}^{n} \frac{1}{c(f_{i + 1})}.
	\]
	Observe that $(a_{n})_{n\in \N}$ is a nondecreasing sequence of real numbers since $\frac{1}{c(f_{i})} \geq 0$ for every $i \in\N\setminus\{1\}$ by Theorem \ref{convindex_values_infinitedimensional}. Moreover, the quasiconvexity of $s$ implies that $s_{n}$ is quasiconvex. In view of Theorem \ref{MainResult-2.2}, for every $n\in \N$, we have $\frac{1}{c(f_{1})} + a_{n} \leq 0$, that is,
	\[
		a_{n}  \leq  - \frac{1}{c(f_{1})}.
	\]
	Observe that $\lim_{n\to\infty} a_{n}$ exists and it is finite since the sequence $(a_{n})_{n\in\N}$ is nondecreasing and bounded from above. Then, $\lim_{n\to\infty} a_{n} \leq - \frac{1}{c(f_{1})}$ so that
	\[
		\sum_{i = 1}^{\infty} \frac{1}{c(f_{i})} = \lim_{n\to\infty} a_{n}  + \frac{1}{c(f_{1})} \leq 0.
	\]
	Hence, $(ii)$ holds in this case.
	
	Conversely, assume that either $(i)$ or $(ii)$ holds. If $(i)$ holds, then $s$ is quasiconvex trivially. Suppose that $(ii)$ holds. Without loss of generality, since $(s_{n})_{n\in\N}$ is absolutely convergent, we can assume that $f_{1}$ is not convex. Note that $f_{i}$ is convex, and consequently, $c(f_{i}) \geq 0$ for each $i \in \N\setminus \{1\}$ by Theorem \ref{convindex_values_infinitedimensional}. Hence, by \eqref{eq:MainResult-2.3-1},
	\begin{align} \label{eq:MainResult-2.3-2}
		\frac{1}{c(f_{1})} + \frac{1}{c(f_{2})} \leq \sum_{i = 1}^{\infty} \frac{1}{c(f_{i})} \leq 0.
	\end{align}
	In view of Theorem \ref{MainResult-2.2}, the convexity of $f_{2}$ and \eqref{eq:MainResult-2.3-2} imply that $f_{1} + f_{2}$ is quasiconvex. Then, by Theorem \ref{TwoFactorCase},
	\begin{equation}\label{eq:MainResult-2.3-3}
		c(f_{1}) + c(f_{2}) \geq 0.
	\end{equation}
	Next, on $\bigtimes_{i = 3}^{\infty} \X_{i}$, define a function $\overline{\overline{s}}$ by 
	\[
		\overline{\overline{s}}(x_{3}, x_{4}, \ldots) \coloneqq \sum_{i = 3}^{\infty}f_{i}(x_{i}).
	\]
	Clearly, $\overline{\overline{s}}$ is convex, and consequently, $c(\overline{\overline{s}}) \geq 0$ by Theorem \ref{convindex_values_infinitedimensional}.  By \eqref{eq:MainResult-2.3-3}, we have
	\[
		c(f_{1}) + c(f_{2}) +  c(\overline{\overline{s}}) \geq 0.
	\]
	It follows from Theorem \ref{TwoFactorCase} that $s$ is quasiconvex. 
\end{proof}

\section{Naturally quasiconvex conditional risk measures}\label{sec:nqc}

Throughout this section, we fix a probability space $(\Omega, \F, \P)$ and denote by $L^0(\F)$ the set of all $\F$-measurable real-valued random variables on $\Omega$, where two elements are distinguished up to $\P$-almost sure equality. Let $p\in [1,+\infty]$. We denote by $L^p(\F)=\{X\in L^0(\F)\colon \norm{X}_p<+\infty\}$ the Banach space of all $p$-integrable random variables equipped with the $L^p$-norm $X\mapsto \norm{X}_p\coloneqq \E[|X|^p]$ for $p<+\infty$, and $X\mapsto \norm{X}_p=\esssup|X|$ for $p=+\infty$. Furthermore, we assume that $L^{p}(\F)$ is equipped with the topology induced by the $L^p$-norm for $p\in [1,\infty)$, and with the weak$^\ast$ topology $\sigma(L^{\infty}(\F), L^{1}(\F))$ for $p = \infty$.

Let $\G$ be a sub-$\sigma$-algebra of $\F$. We denote by $L^{p}(\G)$ the set of $\G$-measurable elements in $L^p(\F)$ and equip it with the subspace topology of $L^p(\F)$. In particular, the dual space of $L^p(\G)$ is identified by $L^{q}(\G)$, where $q\in[1,+\infty]$ is such that $\frac{1}{p} + \frac{1}{q} = 1$, together with the bilinear mapping $\ip{\cdot,\cdot}\colon L^p(\G)\times L^q(\G)\to\R$ defined by
\[
\ip{X,Y}\coloneqq\E\sqb{XY}.
\] 
Let $L^{p}_{+}(\G) \coloneqq \{X \in L^{p}(\G) \colon \P\{X\geq 0\} =1 \}$ be the cone of positive elements in $L^p(\G)$; the cone $L^q_+(\G)$ is defined similarly. Then, the positive dual cone of $L^{p}_{+}(\G)$ is $L^{q}_{+}(\G)$ since we have
\[
	L^{q}_{+}(\G)= \{Y \in L^{q}(\G) \colon \E[YZ ] \geq 0\text{ for every } Z\in L_{+}^{p}(\G) \}.
\]
All of the equalities and inequalities among random variables are assumed to hold $\P$-almost surely. For every $A\in\F$, the stochastic indicator function of $A$ is denoted by ${\bf 1}_A$; we have ${\bf 1}_{A}(\omega)=1$ if $\omega\in A$, and ${\bf 1}_A(\omega)=0$ if $\omega \not\in A$.

	\subsection{Risk measures}\label{subsec:risk}

	In this section, we briefly discuss risk measures and review their basic properties. Let $p\in [1, \infty]$. In the following, we consider the risk measures which are defined on $L^{p}(\F)$ and take values in the closed linear subspace $L^{p}(\G)$ of its domain. In this setting, a risk measure gives the risk of a financial position possibly at an intermediate time. In other words, we are in the conditional setting. Notice that the case where the risk measure is measured at the present time, namely the static case, can be covered by taking $\G = \{ \emptyset, \Omega \}$ (or a trivial $\sigma$-algebra) as a special case.

Below, we define a conditional risk measure as a functional with the minimal set of properties.

\begin{defn} \label{Definition of Risk Measure}
	A mapping $\rho \colon L^{p}(\F) \to L^{p}(\G)$ is a conditional risk measure if it satisfies the following properties:
	\begin{enumerate}[(i)]
		\item Monotonicity : $X \leq Y$ implies $\rho(X) \geq \rho(Y)$  for every $X,Y \in L^{p}(\F)$.
		\item Quasiconvexity : $\rho( \lambda X + (1-\lambda) Y) \leq \max \{ \rho(X), \rho(Y) \}$ for every $X,Y \in L^{p}(\F)$ and $\lambda\in [0,1]$. 
	\end{enumerate}
\end{defn}

The next definition provides some further properties of conditional risk measures which are of importance. 

\begin{defn}\label{defn:properties risk measures}
	A conditional risk measure $\rho \colon L^{p}(\F) \to L^{p}(\G)$ is called 
	\begin{enumerate}[(i)]
		\item translative if $\rho(X + Z) = \rho(X) - Z$ for every $Z\in L^{p}(\G)$.
		\item local if $\rho(X {\bf 1}_{A}) {\bf 1}_{A} = \rho(X) {\bf 1}_{A}$ for every $X\in L^{p}(\F)$ and $A\in \G$.
		\item convex if $\rho(\lambda X + (1-\lambda)  Y)  \leq \lambda \rho(X) + (1-\lambda) \rho(Y)$ for every $X,Y\in L^{p}(\F)$ and $\lambda\in  [0,1]$.
	\item normalized if $\rho(0)=0$.
	\end{enumerate}
\end{defn}

\begin{rem}\label{rem:alternative locality}
	Locality property holds if and only if $\rho(X {\bf 1}_{A} + U {\bf 1}_{A^{c}} ) = \rho (X ) {\bf 1}_{A} + \rho(U) {\bf 1}_{A^{c}}$ for every $X, U \in L^{p}(\F)$ and $A\in \G$.
\end{rem}
Translativity captures the following idea. Suppose that $X, Z$ are financial positions such that worth of $X$ and $Z$ are fully known at the times $T_{1}$ and $T_{2}$, respectively, where $T_{1} < T_{2}$. Suppose further that $\rho$ is a risk measure which measures risk of a financial position at time $T_{1}$ whose value is fully known at time $T_{2}$. Then, one can expect that at time $T_{1}$ the risk of holding $X+Z$ is equal to the risk of holding $X$ and liquidating $Z$. 

Locality means that the events that will not happen in the future have no contribution to the value of risk. Note that this is exactly having a tree-like probabilistic setting as in Figure \ref{Figure: Tree} when $\F$ is finitely generated.

Suppose $X$ and $Y$ are two financial positions. One can diversify by investing some fraction $\lambda\in [0,1]$ of the resources on $X$ and the remaining on $Y$. In other words, one can diversify by choosing to invest on the portfolio $\lambda X + (1-\lambda) Y$ instead of taking position in only one of the assets. The convexity property exactly means that the risk of the diversified position is less than or equal to the weighted sum of individual risks with weights equal to the fractions of the assets in the diversified portfolio. 

As discussed earlier, quasiconvexity is also considered as a formulation of diversification. Under quasiconvexity, for each possible scenario, risk of the diversified portfolio is less than or equal to the risk of the riskier asset. Although both convexity and quasiconvexity refers to the same concept, it is clear that convexity is a more conservative property to capture the concept of diversification.


\subsection{Natural quasiconvexity}\label{NQC defn}

The two properties which are of the utmost importance for the remaining part are natural quasiconvexity and $\star$-quasiconvexity. These properties are studied in \citet{Helbig}, \citet{Tanaka}, \citet{Kuroiwa}, and \citet{Jeyakumar} for vector-valued and set-valued functions on general vector spaces. We study them within the framework of conditional risk measures. Let us start with the definition of a naturally quasiconvex conditional risk measure.

\begin{defn}\label{StarQC defn}
	A conditional risk measure $\rho : L^{p}(\F) \to L^{p}(\G)$ is called naturally quasiconvex if for every $X, Y\in L^{p}(\F)$ and $\lambda\in[0,1]$ there exists $\mu \in [0,1]$ such that 
	\[
		\rho(\lambda X + (1-\lambda) Y) \leq \mu \rho(X) + (1-\mu)\rho(Y).
	\]
\end{defn}

Clearly, natural quasiconvexity is a property which is stonger than quasiconvexity but weaker than convexity. Next, we give definition of $\star$-quasiconvex risk measure.  

\begin{defn}\label{def:star-quasiconvex}
	A conditional risk measure $\rho : L^{p}(\F) \to L^{p}(\G)$ is called $\star$-quasiconvex if for every $Z^{*}\in L^{q}_{+}(\G)$, the function	$X \mapsto \E [\rho(X) Z^{*}]$ is quasiconvex on $L^{p}(\F)$.
\end{defn}

Next, we show the equivalence between natural quasiconvexity and $\star$-quasiconvexity. A set-valued version of the following theorem for abstract topological vector spaces can be found in \citet[Proposition 2.1, Theorem 2.1]{Kuroiwa}.

\begin{thm}\label{StarQC Equivalent to NQC}
	A conditional risk measure $\rho : L^{p}(\F) \to L^{p}(\G)$ is naturally quasiconvex if and only if it is $\star$-quasiconvex. 
\end{thm}
\begin{proof}
	Assume that $\rho$ is naturally quasiconvex and let $Z^{*}\in L^{q}_{+}(\G)$. Then, for every $X, Y\in L^{p}(\F)$ and $\lambda \in [0,1]$ there exists $\mu\in[0,1]$ such that
	\[
		\E[\rho(\lambda X + (1- \lambda) Y) Z^{*}] \leq \mu \E[\rho(X) Z^{*}] + (1-\mu) \E[\rho(Y) Z^{*}] .
	\]
	Clearly,
	\[
		\mu \E[\rho(X) Z^{*}] + (1-\mu) \E[\rho(Y) Z^{*}] \leq \max \{\E[\rho(X) Z^{*}], \E[\rho(Y) Z^{*}] \}.
	\]
	Hence,
	\[
		\E[\rho(\lambda X + (1- \lambda) Y) Z^{*}] \leq \max \{\E[\rho(X) Z^{*}], \E[\rho(Y) Z^{*}] \}.
	\]
	Conversely, assume that $\rho$ is $\star$-quasiconvex and assume to the contrary that $\rho$ is not naturally quasiconvex. Then, there exists $X,Y\in L^{p}(\F)$, $\lambda \in [0,1]$ such that for every $\mu\in [0,1] $ 
	\begin{equation}\label{eq:NQC-starQC-1}
		\rho(\lambda X + (1-\lambda) Y ) > \mu \rho(X) + (1-\mu) \rho(Y).
	\end{equation}
	Clearly, \eqref{eq:NQC-starQC-1} is equivalent to
	\[
		\mu \rho(X) + (1-\mu)\rho(Y)\not\in \rho(\lambda X + (1-\lambda) Y) + L^{p}_{+}(\G).
	\]
	Let us define two sets $A\coloneqq \{\mu\rho(X) + (1-\mu) \rho(Y)  \colon \mu\in [0,1]\}$ and $B \coloneqq \rho(\lambda X + (1-\lambda) Y) + L^{p}_{+}(\G)$. Observe that $A$ is compact in $L^{p}(\G)$, $B$ is closed convex  in $L^{p}(\G)$ and $A\cap B$ = $\emptyset$. Hence, in view of strict separation theorem, there exists $Z^{*}\in L^{q}(\G)$ such that 
	\begin{equation}\label{eq:NQC-starQC-2}
		\sup_{U\in A} \E[Z^{*} U] < \inf_{U\in B} \E[Z^{*} U].
	\end{equation}
	Notice that
	\[
		\sup_{U\in A} \E[Z^{*} U]  = \sup_{\mu\in[0,1] } \big\{\mu \E[Z^{*}\rho(X)] + (1-\mu)\E[Z^{*} \rho(Y)]  \big\}
		=\max \big\{\E[Z^{*}\rho(X)   ], \E[Z^{*} \rho(Y)] \big\}.
	\]
	On the other hand,
	\[
		\inf_{U\in B} \E[Z^{*} U] = \E[Z^{*} \rho(\lambda X + (1-\lambda) Y)]  + \inf_{U^{'}\in L^{p}_{+}(\G)} \E[Z^{*} U^{'}].
	\]
	Moreover, since $\E[U Z^{*}] < 0$ if $Z^{*}\in L_{+}^{q}(\G)$, we have 
	\[
		\inf_{U^{'}\in L^{p}_{+}(\G)} \E[Z^{*} U^{'}] = \begin{cases}
			0 & \text{if } Z^{*}\in L^{q}_{+}(\G),\\
			-\infty & \text{if }Z^{*}\not\in L^{q}_{+}(\G).
		\end{cases}
	\]
	Hence,
	\begin{equation}\label{eq:NQC-starQC-4}
		\inf_{U\in B} \E[Z^{*} U] = \begin{cases}
			\E[Z^{*} \rho(\lambda X + (1-\lambda) Y)]   & \text{if } Z^{*}\in L^{q}_{+}(\G),\\
			-\infty & \text{if }Z^{*}\not\in L^{q}_{+}(\G).
		\end{cases}
	\end{equation}
	In view of \eqref{eq:NQC-starQC-2} and \eqref{eq:NQC-starQC-4}, $Z^{*}\in L^{q}_{+}(\G)$. Therefore,
	\[
		\max \big\{\E[Z^{*}\rho(X)   ], \E[Z^{*} \rho(Y)] \big\} < \E[Z^{*} \rho(\lambda X + (1-\lambda) Y)],   
	\]
	which contradicts the $\star$-quasiconvexity assumption.  
\end{proof}

	
\subsection{Relationship between convexity and natural quasiconvexity}\label{subsec:rel}

		In light of the discussion of the previous section, natural quasiconvexity and convexity are closely related. Indeed, we will prove that they are equivalent properties for risk measures under some mild assumptions. To be more precise, we will require that the risk measure is, in some sense, non-constant, lower semicontinuous and local. In the example below, the role of being non-constant and lower semicontinuous can be clearly seen since we make use of Theorem \ref{MainResult-2.2}, which works under these properties. On the other hand, locality is not explicitly emphasized but can be understood from the tree-like like structure of the probability space as can be seen in Figure \ref{Figure-Example 1}. 
	
	\begin{example} \label{Finitely Generated Example}
		Let $\Omega = \{\omega_{1}, \ldots, \omega_{10}  \}$,  $X:\Omega\to\R$ be an $\F$-measurable random variable and $\rho : L^{p}(\F)\to L^{p}(\G)$ be a local conditional risk measure. We take
		\[
			\F = 2^{\Omega},\quad 
			\G = \sigma\big(\big\{ \{\omega_{1},\omega_{2},\omega_{3},\omega_{4}\} , \{\omega_{5},\omega_{6},\omega_{7}\}, \{\omega_{8},\omega_{9},\omega_{10}\}     \big\}\big).
		\]
		
		\tikzstyle{level 1}=[level distance=3.5cm, sibling distance=2cm]
		\tikzstyle{level 2}=[level distance=3.5cm, sibling distance=0.6cm]
		
		\tikzstyle{bag} = [text width=4em, text centered]
		\tikzstyle{end} = [circle, minimum width=3pt,fill, inner sep=0pt]
		
		\begin{center}
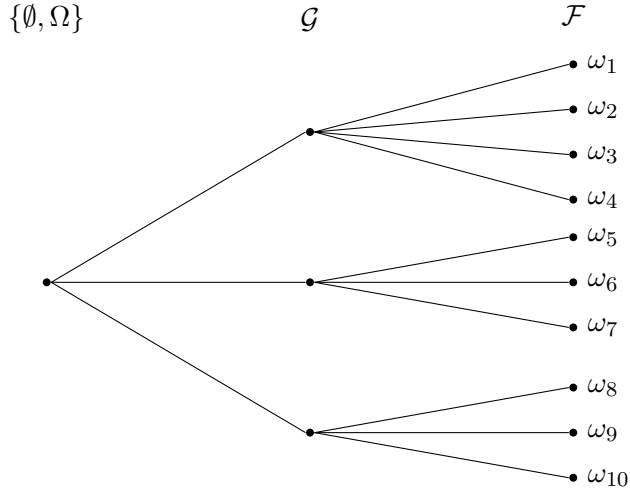

			
			\begin{tikzpicture}[grow=right, sloped]\label{Figure-Example 1}
				\node(Root)[end] {}
				child {
					node[end] {}
					child {
						node[end, label=right:
						{$\omega_{10}$}] {}
					}
					child {
						node[end, label=right:
						{$\omega_{9}$}] {}
					}
					child {
						node[end, label=right:
						{$\omega_{8}$}] {}
					}
				}
				child {
					node[end] {}
					child {
						node[end, label=right:
						{$\omega_{7}$}] {}
					}
					child {
						node[end, label=right:
						{$\omega_{6}$}] {}
					}
					child {
						node[end, label=right:
						{$\omega_{5}$}] {}
					}
				}
				child {
					node[end] {}
					child {
						node[end, label=right:
						{$\omega_{4}$}] {}
					}
					child {
						node[end, label=right:
						{$\omega_{3}$}] {}
					}
					child {
						node[end, label=right:
						{$\omega_{2}$}] {}
					}
					child {
						node[end, label=right:
						{$\omega_{1}$}] {}
					}
				};
				\path (Root)  ++(0cm,3.5cm) node{$\{ \emptyset, \Omega \}$};
				\path (Root-3)  ++(0,1.5cm) node{$\G$};
				\path (Root-3-4)  ++(0cm,0.65cm) node{$\F$};
			\end{tikzpicture}
			\captionof{figure}{Relationship between convex risk measures and natural quasiconvex risk measures when $\F$ is finitely generated.}
			\label{Figure: Tree}
		\end{center}
		
		In this setting, $L^{p}(\F) \cong \R^{4}\times\R^{3}\times\R^{3} \cong \R^{10}$ and $L^{p}(\G) \cong \R^{3}$. Since $X\in L^{p}(\F)$, it can be identified by the vector $x=
			 (X(\omega_{1}), \ldots, X(\omega_{10}))^{\mathsf{T}}\in \R^{10}$. For such $x\in\R^{10}$, we write $x_{1:4}=(X(\omega_1,\ldots,X(\omega_4)))^{\mathsf{T}}$, $x_{5:7}=(X(\omega_5),X(\omega_6),X(\omega_7))^{\mathsf{T}}$, $x_{8:10}=(X(\omega_8),X(\omega_9),X(\omega_{10}))^{\mathsf{T}}$. Furthermore, thanks to locality, $\rho$ can be seen as a vector-valued function from $\R^{10}$ into $\R^3$ if we write
		$\rho(x) =(\rho_{1}(x_{1:4}), \rho_{2}(x_{5:7}), \rho_{3}(x_{8:10}))^{\mathsf{T}}\in \R^{3}$
		for $x\in\R^{10}$, where $\rho_1(x_{1:4})$ is the constant value of $\rho(X)$ on $\{\omega_1,\ldots,\omega_4\}$, and $\rho_2(x_{5:7}), \rho_3(x_{8:10})$ are defined similarly. 
		
		Assume that $\rho$ is naturally quasiconvex and $\rho_{i}$ is lower semicontinuous, non-constant for each $i \in \{ 1,2,3 \}$. In this finite-dimensional setting, clearly, $\rho$ is lower semicontiuous  if and only if $\rho_{1},\rho_2,\rho_3$ are lower semicontinuous. It follows from Theorem \ref{StarQC Equivalent to NQC} that $\rho$ is $\star$-quasiconvex. Then, by Definition \ref{StarQC defn}, for every $w=(w_{1},w_{2},w_{3})\in\R^{3}_{+}$,
		\begin{equation}\label{eq:Finitely Generated Example-1}
			x\mapsto w_{1}\rho_{1}(x_{1:4}) + w_{2}\rho_{2}(x_{5:7}) + w_{3}\rho_{3}(x_{8:10})
		\end{equation}
		is quasiconvex on $\R^{10}$. Depending on the nonnegativity of the convexity indices of $\rho_1,\rho_2,\rho_3$, there are eight possible cases.
		\begin{enumerate}[(a)]
			\item $c(\rho_{1}) \geq  0$, $c(\rho_{2}) \geq 0$ and $c(\rho_{3}) \geq 0$.
			\item $c(\rho_{1}) \geq 0$, $c(\rho_{2}) \geq 0$ and $c(\rho_{3}) < 0$. Then, it follows from Theorem \ref{convindex_values_infinitedimensional} that $\rho_{1}, \rho_{2}$ are convex and $\rho_{3}$ is not convex. In view of \eqref{eq:Finitely Generated Example-1} and Theorem \ref{MainResult-2.2}, we have
			\begin{equation}\label{eq:Finitely Generated Example-2}
				\frac{1}{c(w_{1}\rho_{1})} + \frac{1}{c(w_{2}\rho_{2})} + \frac{1}{c(w_{3}\rho_{3})} \leq 0.
			\end{equation}
			It follows from Lemma \ref{Homegenity Convexity Index} together with \eqref{eq:Finitely Generated Example-2} that
			\begin{equation}\label{eq:Finitely Generated Example-3}
				\frac{1}{c(w_{1}\rho_{1})} + \frac{1}{c(w_{2}\rho_{2})} + \frac{1}{c(w_{3}\rho_{3})}= \frac{w_{1}}{c(\rho_{1})} + \frac{w_{2}}{c(\rho_{2})} + \frac{w_{3}}{c(\rho_{3})}
			\end{equation}
			Notice that, since $\rho_{i}$ is non-constant for each $i \in \{1,2,3\}$, $c(\rho_{i}) < +\infty$ for each $i \in\{1,2,3\}$. Let us set $(w_{1},w_{2}, w_{3}) = (c(\rho_{1}),c(\rho_{2}), -c(\rho_{3}) )$ in \eqref{eq:Finitely Generated Example-3}. Then,
			\[
				\frac{1}{c(w_{1}\rho_{1})} + \frac{1}{c(w_{2}\rho_{2})} + \frac{1}{c(w_{3}\rho_{3})} =   1,
			\]
			which is in contradiction with \eqref{eq:Finitely Generated Example-2}. Hence, this case is eliminated.
			\item $c(\rho_{1}) \geq 0$, $c(\rho_{2}) < 0$ and $c(\rho_{3}) < 0$. Then, $\rho_{2}$ and $\rho_{3}$ are not convex by Theorem \ref{convindex_values_infinitedimensional}. On the other hand, since \eqref{eq:Finitely Generated Example-1} holds, it follows from Theorem \ref{MainResult-2.2} that at most one of $\rho_{1}$, $\rho_{2}$ or $\rho_{3}$ is not convex. We reached a contradiction, hence this case is eliminated. 
			\item $c(\rho_{1}) < 0$, $c(\rho_{2}) < 0$ and $c(\rho_{3}) < 0$. By a similar argument as in case (c), we can eliminate this case.  
		\end{enumerate}
		Observe that the remaining four cases are symmetric either with the case (b) or (c). Therefore, we are left with the case (a) and $\rho_{1},\rho_{2},\rho_{3}$ are convex. Hence, $\rho$ is convex. 
	\end{example}
	In  Example \ref{Finitely Generated Example}, we need non-constancy of $\rho_{i}$ for each $i\in \{1,2,3\}$ to draw the connection between natrual quasiconvexity and convexity. We need a similar assumption to generalize Example \ref{Finitely Generated Example}, which we state next.
	\begin{assumption}\label{as:Convexity and NQC- general case}
		For every $A\in \G$ with $\P(A) > 0$, the function $X\mapsto \E[\rho(X)1_{A}]$ is non-constant on $L^p(\F)$. 
	\end{assumption}
	To motivate this assumption, let us give an example of a quasiconvex risk measure which satisfies Assumption \ref{as:Convexity and NQC- general case}.
	\begin{example}\label{ex:assumption3.3}
		Consider the certainty equivalent risk measure 
		\[
			\rho(X) = \ell^{-1}\E[\ell(-X)|\G],\quad X\in L^p(\F),
		\]
		where $\ell\colon \R\to\R$ is a continuous increasing loss function. Let $x_{1}, x_{2} \in \R$ such that $x_{1}\neq x_{2}$. Observe that
		\begin{align*}
			&\rho(x_{1}) = \ell^{-1}\E[\ell(-x_{1})|\G] = \ell \big( \ell^{-1}(-x_{1})\big) = -x_{1},\\
			&\rho(x_{2}) = \ell^{-1}\E[\ell(-x_{2})|\G]  = \ell \big( \ell^{-1}(-x_{2})\big) = -x_{2}.
		\end{align*}
		Hence, 
		\[
			\E[\rho(x_{1}) {\bf 1}_{A}] - \E[\rho(x_{2}) {\bf 1}_{A}] 
			=\E\big[\big(\rho(x_{1}) -\rho(x_{2})\big) {\bf 1}_{A}\big] 
			= \E[(x_{2} - x_{1}) {\bf 1}_{A}]
			= (x_{2} - x_{1})\P(A) 
			> 0,
		\]
		where the strict inequality holds since $x_{1} \neq x_{2}$ and $\P(A) >0$. Hence, $X \mapsto \E[\rho(X){\bf 1}_{A}]$ is non-constant.
	\end{example}
	We would like to relate Assumption \ref{as:Convexity and NQC- general case} with sensitive risk measures, which are discussed in \citet{Follmer-Penner}, \citet{Cheridito} and \citet{Kloppel}. First let us give the definition of a sensitive risk measure.
	\begin{defn}
		$\rho$ is called sensitive if $\P\big\{\rho(-\epsilon {\bf 1}_{A}) > 0\big\} > 0$ holds for every $\epsilon > 0$ and $A\in \F$ such that $\P(A) > 0$.  
	\end{defn}
	\begin{rem}\label{rem:Sensitivity}
	A stronger property for $\rho$ than sensitivity is the following: $A \subseteq \{ \rho(-\epsilon {\bf 1}_{A}) > 0\}$ holds for every $A\in \F$ and $\epsilon> 0$.
	\end{rem}
	The connection between Assumption \ref{as:Convexity and NQC- general case} and sensitive risk measures is given in the next proposition. 
	
	\begin{prop}
		Assume that $\rho$ is local and normalized. If $\rho$ is sensitive, then $\rho$ satisfies Assumption \ref{as:Convexity and NQC- general case}.
	\end{prop}

	\begin{proof}
 		We prove the contrapositive. Assume that Assumption \ref{as:Convexity and NQC- general case} does not hold. Then, there exists $B\in \G$ with $\P(B) > 0$ such that $X\mapsto \E[\rho(X){\bf 1}_{B}]$ is constant. In particular, for every $X\in L^{p}(\F)$, we have
 		\[
 			\E[\rho(X){\bf 1}_{B}] = \E[\rho(0){\bf 1}_{B}] = 0.
 		\]
 By the locality of $\rho$, this implies that $\E[\rho(X {\bf 1}_{B})] = 0$ for every $X\in L^p(\F)$. In particular, 
 we have
 		$\E[\rho(-\epsilon {\bf 1}_{B})]  = 0$. 
 	Furthermore, by monotonicity, we have 
 	$\rho(-\epsilon {\bf 1}_{B}) \geq \rho(0) = 0$. Hence,
 		$\rho(-\epsilon {\bf 1}_{B})= 0$. 
 	Therefore, $\rho$ is not sensitive, which completes the proof.
	\end{proof}

	As a suitable semicontinuity notion for naturally quasiconvex conditional risk measures, we recall the following definition: a conditional risk measure $\rho$ is called $\star$-lower semicontinuous if and only if the function $X\mapsto \E[\rho(X)Z^\ast]$ is lower semicontinuous on $L^p(\F)$ for every $Z^\ast\in L^q_+(\G)$.
	
	The next theorem is a generalization of Example \ref{Finitely Generated Example}.
	
	\begin{thm} \label{Thm:NQC is equal to Convexity}
		Suppose that $\F$ is non-trivial with respect to $\P$, that is, there exists $A_0\in\F$ such that $\P(A_0)\in (0,1)$. Assume that $\rho \colon L^{p}(\F) \to L^{p}(\G)$ is a $\star$-lower semicontinuous local conditional risk measure which satisfies Assumption \ref{as:Convexity and NQC- general case}. Then, $\rho$ is convex if and only if it is naturally quasiconvex. 
	\end{thm}

	\begin{proof}
		Trivially, convexity implies natural quasiconvexity. Conversely, assume that $\rho$ is naturally quasiconvex. Let $X,Y\in L^p(\F)$ and $\lambda\in[0,1]$. We show that $\rho((\lambda X + (1-\lambda) Y)) \leq \lambda \rho(X) + (1-\lambda) \rho(Y)$. By the well-known characterization of almost sure ordering of random variables, this is equivalent to having
		\begin{equation}\label{eq:NQC}
		\E[\rho((\lambda X + (1-\lambda) Y)) {\bf 1}_{A}] \leq  \E[(\lambda \rho(X) + (1-\lambda) \rho(Y) ){\bf 1}_{A}]
		\end{equation}
		for every $A\in\F$ with $\P(A) > 0$; see, e.g., \citet[Remark 2.3(d)]{Cinlar}. We claim that it is enough to show \eqref{eq:NQC} only for every $A\in\F$ with $\P(A)\in (0,1)$. Indeed, if $B\in\F$ with $\P(B)=1$, then 
		\begin{align*}
		&\E[\rho((\lambda X + (1-\lambda) Y)) {\bf 1}_{B}]=\E[\rho((\lambda X + (1-\lambda) Y)) {\bf 1}_{A_0}]+\E[\rho((\lambda X + (1-\lambda) Y)) {\bf 1}_{A_0^c}],\\
		& \E[(\lambda \rho(X) + (1-\lambda) \rho(Y) ){\bf 1}_{B}]= \E[(\lambda \rho(X) + (1-\lambda) \rho(Y) ){\bf 1}_{A_0}]+\E[(\lambda \rho(X) + (1-\lambda) \rho(Y) ){\bf 1}_{A_0^c}],
		\end{align*}
		where $A_0\in\F$ is an arbitrary event such that $\P(A_0)\in (0,1)$. Hence, having \eqref{eq:NQC} for $A=A_0$ and $A=A_0^c$ implies that \eqref{eq:NQC} holds for $A=B$ as well.
		
		Let $A \in \G$ such that $\P(A) \in (0,1)$. Let us take $Z^\ast=z_1{\bf 1}_A+z_2{\bf 1}_{A^c}$ for some arbitrary $z_1,z_2>0$. Then, by Theorem \ref{StarQC Equivalent to NQC}, the function $X \mapsto \E [\rho(X) Z^\ast]$ is quasiconvex on $L^{p}(\F)$. By the definition of $Z^\ast$ and the locality of $\rho$, we have
		\[
		\E [\rho(X) Z^\ast]=z_1\E[\rho(X){\bf 1}_A]+z_2\E[\rho(X){\bf 1}_{A^c}]=z_1\E[\rho(X{\bf 1}_A){\bf 1}_A]+z_2\E[\rho(X{\bf 1}_{A^c}){\bf 1}_{A^c}]
		\]
		for every $X\in L^p(\F)$. Let us consider the probability space $(A,\F_{A},\P_A)$, where $\F_A$ is the trace of $\F$ on $A$, and $\P_A$ is the conditional probability given $A$ considered on $\F_A$. The probability space $(A^c,\F_{A^c},\P_{A^c})$ is defined similarly. For each $X\in L^p(\F)=L^p(\Omega,\F,\P)$, we may write $X=X_A+X_{A^c}$, where $X_A\in L^p(A,\F_A,\P_A)$, $X_{A^c}\in L^p(A^c,\F_{A^c},\P_{A^c})$ denote the restrictions of $X$ on $A, A^c$, respectively. It follows that, $L^{p}(\Omega,\F,\P) \cong  L^{p}(A,\F_A,\P_A)  \times   L^{p}(A^{c},\F_{A^c},\P_{A^c})$. Moreover, $L^{p}(A,\F_A,\P_A)$, $L^{p}(A^{c},\F_{A^c},\P_{A^c})$ can be naturally imbedded into $L^p(\Omega,\F,\P)$; we rely on these embeddings below without introducing additional notation. Let us define $f_1\colon L^p(A,\F_A,\P_A)\to\R$, $f_{2}\colon L^p(A^c,\F_{A^c},\P_{A^c})\to\R$ by
		\[
		f_1(X_1)\coloneqq \E[\rho(X_1){\bf 1}_A],\ X_1\in L^p(A,\F_A,\P_{A});\quad f_{2}(X_2)\coloneqq \E[\rho(X_2){\bf 1}_{A^c}],\ X_2\in L^p(A^c,\F_{A^c},\P_{A^c}).
		\]
		Hence, $\E[\rho(X)Z^\ast]=z_1f_1(X_{A})+z_2f_{2}(X_{A^c})$ for every $X\in L^p(\Omega,\F,\P)$. It follows that the decomposable sum $(X_1,X_2)\mapsto z_1f_1(X_1)+z_2 f_{2}(X_2)$ is quasiconvex on $L^{p}(A,\F_A,\P_A)  \times   L^{p}(A^{c},\F_{A^c},\P_{A^c})$.
		
		Note that $f_{1},f_{2}$ are non-constant functions by Assumption \ref{as:Convexity and NQC- general case}. In particular, $c(f_1)<+\infty$, $c(f_2)<+\infty$. These functions are also lower semicontinuous since $\rho$ is $\star$-lower semicontinuous. We consider the following four cases:
		\begin{enumerate}[(a)]
			\item $c(f_{1}) \geq 0$, $c(f_{2}) \geq 0$.
			\item $c(f_{1}) \geq  0$, $c(f_{2}) < 0$. Then, it follows from Theorem \ref{convindex_values_infinitedimensional} that $f_{1}$ is convex and $f_{2}$ is not convex. In view of Lemma \ref{Homegenity Convexity Index} and Theorem \ref{MainResult-2.2}, we have
			\begin{equation}\label{eq:NQC is equal to Convexity-4}
				\frac{z_{1}}{c(f_{1})} + \frac{z_{2}}{c(f_{2})}=\frac{1}{c(z_{1}f_{1})} + \frac{1}{c(z_{2}f_{2})} \leq 0.
			\end{equation}
			Let us set $(z_{1},z_{2}) = (c(f_{1}),-\frac{1}{2}c(f_{2}))$. Then, $\frac{1}{c(z_{1}f_{1})} + \frac{1}{c(z_{2}f_{2})} = \frac{1}{2}$,	which contradicts \eqref{eq:NQC is equal to Convexity-4}. Hence, this case is eliminated.
			\item $c(f_{1}) < 0$, $c(f_{2}) \geq 0$. This case is symmetric with case (b), hence eliminated. 
			\item $c(f_{1}) < 0$, $c(f_{2}) < 0$. Then, neither $f_{1}$ nor $f_{2}$ is convex by Theorem \ref{convindex_values_infinitedimensional}. On the other hand, by taking $z_1=z_2=1$, the function $(X_1,X_2)\mapsto f_1(X_1)+f_2(X_2)$ is quasiconvex. Then, by Theorem \ref{MainResult-2.2}, at least one of $f_{1}$ or $f_{2}$ is convex, a contradiction. Hence, this case is eliminated.
		\end{enumerate}
		Summing up, we are left with the case (a), that is, $f_1, f_2$ are convex. In particular, $X\mapsto f_1(X_A)=\E[\rho(X_A){\bf 1}_A]=\E[\rho(X){\bf 1}_A]$ is convex so that
		\[
		\E[\rho((\lambda X + (1-\lambda) Y)) {\bf 1}_{A}] \leq  \lambda\E[\rho(X){\bf 1}_A]+(1-\lambda)\E[\rho(Y){\bf 1}_A]= \E[(\lambda \rho(X) + (1-\lambda) \rho(Y) ){\bf 1}_{A}].
		\]
		Hence, \eqref{eq:NQC} follows, which completes the proof.
	\end{proof}

	
	\subsection{Relationship between convexity and natural quasiconvexity on $L^{2}$}\label{subsec:L2}
	
	We turn our attention to the case $\rho \colon L^{2}(\F) \to L^{2}(\G)$. As a preparation, let us review some basic properties of $L^{2}(\F)$ and $L^{2}(\G)$.
	
	The sets $L^{2}(\F)$ and $L^{2}(\G)$ are separable Hilbert spaces with respect to the inner product 
	\begin{equation} \label{defn:inner product}
		\langle X , Y \rangle = \E[X Y].
	\end{equation}
	It is a well-known result that every separable Hilbert space has a countable orthonormal basis. Thus, one can find countable orthonormal bases for $L^{2}(\F)$ and $L^{2}(\G)$. Moreover, $L^{2}(\G)$ is a closed subspace of $L^{2}(\F)$ and we have
	\[
		L^{2}(\F) = L^{2}(\G) \oplus  L^{2}(\G)^{\perp},
	\]
	where $L^{2}(\G)^{\perp} = \{ X\in L^{2}(\F) \colon \langle X, Y \rangle = 0\text{ for every }Y\in L^{2}(\G) \}$ and $\oplus$ denotes internal direct sum for subspaces. We also define, for every $X\in L^{2}(\F)$, the orthogonal complement of $X$ by
	\[
		X^{\perp} \coloneqq X - \E[X|\G].
	\]
	Let $\{ e^{i} \colon i\in I\}$ be a countable orthonormal basis of $L^{2}(\G)$. It is orthonormal in the sense that
	\begin{enumerate}[(i)]
		\item $\norm{e^{i}} = 1$ for each $i\in I$,
		\item $\langle e^{i}, e^{j} \rangle = 0$ for each $i,j\in I$ such that $i \neq j$,
	\end{enumerate}
	where $\norm{\cdot}$ is the $L^{2}$-norm induced by the inner product in \eqref{defn:inner product}. Furthermore, for every $Y\in L^{2}(\G)$, we have
	\[
		Y = \sum_{i\in I} \langle Y, e^{i} \rangle e^{i}.
	\]
	
	A natural question is whether or not there is an extension of $\{ e^{i} \colon i\in I\}$ to an orthonormal basis of $L^{2}(\F)$. Luckily, one can find such an extension and a proof of this result can be found in \citet[Theorem 6.29]{Hunter}. Since $L^2(\F)$ is separable, such basis is necessarily countable. In light of these arguments, there exists $I^{\prime}\supseteq I$ such that  $\{ e^{i} \colon i\in I^{\prime} \}$ is a countable orthonormal basis of $L^{2}(\F)$.
	
	Let us come back to the discussion of this section. We will work under the following assumption. 
	\begin{assumption}\label{ass:structure}\rm
	     The reference probability space $\Omega$ is partitioned into (at most) countably many sets, i.e.,
	     $\Omega=\bigcup_{i\in \bbN} \Omega_i $ for some $\bbN\subseteq\N$ and disjoint sets $\Omega_i\in\F$, $i\in\bbN$,
	     and there exist closed subspaces $H_i, i\in \bbN$, of $L^2(\cG)$ such that 
	     \[
	     H_i \subseteq \Lambda_i:= \{ X \in L^2(\cF): X{\bf 1}_{\Omega_i^c}=0\},\quad i \in \bbN,
	     \]
	     and
	     \[
	        L^2(\cG)= \bigoplus_{i\in \bbN} H_i\coloneqq \cb{\sum_{i\in\bbN}X_i\colon \sum_{i\in\bbN}\norm{X_i}^2<+\infty,\ \forall i\in\bbN\colon X_i\in H_i}.
	     \]
	\end{assumption}
	

Notice that Assumption \ref{ass:structure} also ensures that $L^2(\cF)=\bigoplus_{i\in \bbN} \Lambda_i$. 
Moreover, under this assumption, the spaces $\Lambda_i$, $i\in\bbN$, are orthogonal. In fact, for every $i,j\in \bbN$ with $i\neq j$, and $X_i\in \Lambda_i,X_j \in \Lambda_j$ we have
$\langle X_i,X_j \rangle=\E[ X_i X_j] = 0$ since $X_i{\bf 1}_{\Omega_j}=X_j{\bf 1}_{\Omega_i}=0$.


Under Assumption \ref{ass:structure}, for each $i\in\bbN$, the closed subspace $H_i$ and its orthogonal complement $H_i^{\bot,\Lambda_i}$ in $\Lambda_i$ are separable Hilbert spaces; hence each admits a (at most) countable orthonormal basis. In the following, for each $i \in \bbN$, we denote by $(e_k^i)_{k\in {\bbN_i}}$ a countable orthonormal basis for $H_i$, and by $(\beta_k^i)_{k\in \bbM_i}$ a countable orthonormal basis for $H_i^{\bot,\Lambda_i}$, where $\bbN_i,\bbM_i\subseteq \bN$. Hence, 
 \begin{itemize}
 \item[-] for each $i\in \bbN$, $\Theta_i:=((e_k^i)_{k\in \bbN_i}, (\beta_k^i)_{k\in \bbM_i})$ is an orthonormal basis for $\Lambda_i$.
 
  \item[-] $(e_k^i)_{k\in \bbN_i,\, i\in \bbN}$ is an orthonormal basis for $L^2(\cG)$ while $(\beta_k^i)_{k\in \bbM_i,\, i\in \bbN}$ is an orthonormal basis for $L^2(\cG)^\bot$.
    
    \item[-] $\Theta:=(\Theta_i)_{i\in\bbN}$ is an orthonormal basis for $L^2(\cF)$. 
\end{itemize}

	Before proceeding further, let us give some examples where Assumption \ref{ass:structure} is satisfied.
	\begin{example}\label{ex:finite-dim}
	     Assume that $\Omega=\left\{ \omega_1,\dots,\omega_{10} \right\}$ with the power set
$\cF=2^\Omega$, and $\mathbb{P}$ is the (discrete) uniform distribution on $\Omega$. Moreover, assume that $\cG= \sigma( A_1,A_2,A_3)$, where
\[
A_1:=\left\{\omega_1,\omega_2,\omega_3,\omega_4 \right\}, \qquad 
A_2:=\left\{\omega_5,\omega_6,\omega_7 \right\}, \qquad
A_3:=\left\{\omega_8,\omega_9,\omega_{10} \right\}.
\]
Notice that $A_1,A_2,A_3$ form a partition of $\Omega$.
For $i\in\{1,2,3\}$, define
\[
H_i:=
 \spn\cb{ \frac{{\bf 1}_{A_i}}{\|{\bf 1}_{A_i}\|}},\qquad \Lambda_i:= \cb{ X \in L^2(\cF)\colon  X{\bf 1}_{A_i^c} =0}.
\]
Note that $\Lambda_1\cong \R^4$, $\Lambda_2\cong \R^3$, $\Lambda_3\cong \R^3$.
For simplicity, as in Example \ref{Finitely Generated Example}, we identify an element in 
$L^2(\cF)$ with a vector in $\R^{10}$, e.g., we write
$
  {\bf 1}_{A_1} = ( 1, 1,1, 1,0,0,0,0,0,0,0).
$
Clearly, $H_i\subseteq \Lambda_i$ for each $i\in\{1,2,3\}$; moreover, the spaces $H_1, H_2, H_3$ are orthogonal one to each other, the same holds for $\Lambda_1,\Lambda_2,\Lambda_3$. Finally, we have
$L^2(\cG)= H_1 \oplus H_2 \oplus H_3$ and
$L^2(\cF)=\Lambda_1 \oplus \Lambda_2 \oplus \Lambda_3.
$

We introduce the bases for $H_i,H_i^{\bot,\Lambda_i}, i\in\{1,2,3\}$ as follows. For each $i\in\{1,2,3\}$, we set $e^i = \frac{{\bf 1}_{A_i}}{\|{\bf 1}_{A_i}\|}$. 
For $H_1^{\bot,\Lambda_1}\cong\R^3$, we set 
  \begin{align*}
 \beta_1^1&= \frac{\tilde{\beta}_1^1}{\|\tilde{\beta}_1^1\|} \ \ \text{with} \ \ \tilde{\beta}_1^1 = (1,1,-1,-1,0,0,0,0,0,0), \qquad \\
 \beta_2^1&= \frac{\tilde{\beta}_2^1}{\|\tilde{\beta}_2^1\|}  \ \ \text{with} \ \ \tilde\beta_2^1=(1,-1,1,-1,0,0,0,0,0,0),\\
\beta_3^1&= \frac{\tilde{\beta}_3^1}{\|\tilde{\beta}_3^1\|} \ \ \text{with} \ \ \tilde{ \beta}_3^1=(-1,0,0,-1,0,0,0,0,0,0).
  \end{align*}
 For $H_2^{\bot,\Lambda_2}\cong\R^2$, we set
   \begin{align*}
  \beta_1^2&= \frac{\tilde{\beta}_1^2}{\|\tilde{\beta}_1^2\|} \ \ \text{with} \ \  \tilde\beta_1^2= \left(0,0,0,0,-\frac12,-\frac12,1,0,0,0\right) \\
 \beta_2^2&= \frac{\tilde{\beta}_2^2}{\|\tilde{\beta}_2^2\|}  \ \ \text{with} \ \  \tilde\beta_2^2=(0,0,0,0,-1,1,0,0,0,0).
  \end{align*}
  Similarly, for $H_3^{\bot,\Lambda_3}\cong\R^2$, we set
  \begin{align*}
 \beta_1^3&= \frac{\tilde{\beta}_1^3}{\|\tilde{\beta}_1^3\|} \ \ \text{with} \ \  
 \tilde\beta_1^3= \left(0,0,0,0,0,0,0-\frac12,-\frac12,1\right) \\
 \beta_2^3&= \frac{\tilde{\beta}_2^3}{\|\tilde{\beta}_2^2\|} \ \ \text{with} \ \  \tilde\beta_2^3=(0,0,0,0,0,0,0,-1,1,0) .
  \end{align*}
	\end{example}
	
\begin{example}\label{ex:infinit-dim} Assume that $(\Omega,\F,\P)$ is a probability space and $\mathcal{G}=\sigma(\Omega_i\colon i\in\bbN)$. For each $i\in \bN$, let $H_i:= \{ X \in L^2(\cF)\colon X{\bf 1}_{\Omega_i}=a{\bf 1}_{\Omega_i}\text{ for some }a\in\R, X{\bf 1}_{\Omega_i^c}  =0 \}$
   and
$   \Lambda_i:= \{ X \in L^2(\cF): X{\bf 1}_{\Omega_i^c}  =0 \}$; clearly
\medskip 
$H_i,\Lambda_i$ are closed subspaces in 
$L^2(\cF)$ and $H_i\subseteq \Lambda_i$. Hence, they are Hilbert spaces with respect to the norm induced by $L^2(\cF)$. Moreover, we have
\[
    L^2(\cG)= \bigoplus_{i\in\bbN} H_i, \qquad L^2(\cF)=\bigoplus_{i\in\bbN} \Lambda_i.
\]
We emphasize that, similar to Example \ref{ex:finite-dim}, for each $i\in \bN$, $e^i = \frac{{\bf 1}_{\Omega_i}}{\|{\bf 1}_{\Omega_i}\|}$ forms a basis for $H_i$.
\end{example}

\begin{example} \label{ex:borel01}
Assume that $\Omega=[0,1]$ with the Borel $\sigma$-algebra $\F=\B([0,1])$ and the Lebesgue measure $\P=\Leb$. Let $\cG= \B([0,\frac12]) \vee \{(\frac12,1] \}$, which consists of the Borel subsets of $[0,\frac12]$ and their unions with $(\frac12,1]$. Consider $\Omega_1=[0,\frac12]$, $\Omega_2=(\frac12,1]$.
Notice that each $\G$-measurable function is constant on 
$\Omega_2$. Let 
\begin{align*}
H_1&=\{ X\in L^2(\cF): X{\bf 1}_{(\frac12,1]}=0\},\\
H_2&=\{ X\in L^2(\cF): X{\bf 1}_{[0,\frac12]}=0, X{\bf 1}_{(\frac12,1]}=a{\bf 1}_{(\frac12,1]}\text{ for some }a\in\R\}.
\end{align*}
We stress that $H_1$ is infinite-dimensional while $H_2$ has finite dimension.
Moreover, set
\[
\Lambda_1=H_1,\qquad \Lambda_2=\{ X\in L^2(\cF): X{\bf 1}_{[0,\frac12]}=0 \}.
\]
Then a basis for $H_1$ is given by
\[
     e^1_{2k}(x)= 2\sin(4\pi kx){\bf 1}_{[0,\frac12]}(x), \qquad 
	 e^1_{2k-1}(x)=2\cos(4\pi kx){\bf 1}_{[0,\frac12]}(x), \qquad k\in \N,\ x\in[0,1];
\]
while a basis for $H_2$ is given by the unit vector $e^2=\sqrt{2}\cdot {\bf 1}_{(\frac12,1]}$. 
 
Note that $H_1^{\bot,\Lambda_1}=\{0\}$. A basis for $H_2^{\bot,\Lambda_2}$ is given by
 \[
      \beta^2_{2k}(x)= 2\sin(4\pi kx){\bf 1}_{(\frac12,1]}(x), \qquad
     \beta^2_{2k-1}(x)=2\cos(4\pi kx){\bf 1}_{(\frac12,1]}(x), \qquad k\in \N,\ x\in[0,1].
 \]
Under this construction, we have $H_1 =\Lambda_1$, $H_2\subseteq \Lambda_2$ as well as $L^2(\cG)=H_1 \oplus H_2$, $L^2(\cF)=\Lambda_1 \oplus \Lambda_2$.
\end{example}

	We continue by introducing a new locality property that is defined with respect to a given basis. This property precisely says that, when calculating the inner product of $\rho(X)$ with an element of the basis, one can replace $X$ with the sum of its projection on that basis and its orthogonal complement.
	
	\begin{defn}\label{defn:locality w.r.t. basis}
		$\rho$ is called local with respect to the orthonormal basis $(e_k^i)_{k\in \bbN_i,\,i\in \bbN}$ for $L^2(\G)$ if 
		\[
			\langle \rho(X),e_{k_0}^i \rangle =\bigg \langle \rho\Bigg(\sum_{k\in \bbN_i} \langle X,e_{k}^i\rangle e_k^i+ X^{\bot,\Lambda_i}\Bigg),e_{k_0}^i\bigg\rangle
		\]
		for each $X\in L^2(\F)$, $k_0\in\bbN_i$, $i\in\bbN$, where $X^{\bot,\Lambda_i}$ denotes the projection of $X$ onto $H_i^{\bot,\Lambda_i}$, $i\in\bbN$.
	\end{defn}
	
	In the following, we are concerned with the connections between the classical notion of locality in Definition \ref{defn:properties risk measures} and the new notion of locality with respect to a given basis. As a first result, we prove that the former implies the latter. 
	
	\begin{prop}\label{prop:implication}
	Under Assumption \ref{ass:structure}, suppose that $\rho$ is local in the sense of Definition \ref{defn:properties risk measures}. Then, it is local with respect to $(e_k^i)_{k\in \bbN_i,\,i\in \bbN}$.
	\end{prop} 
	
	\begin{proof}
	    Let $X\in L^2(\F)$. For fixed $i\in \bbN$, let us denote by $X^i$ the projection of $X$ on $\Lambda_i$, i.e., $X^i=X {\bf 1}_{\Omega_i}$. As an element of $\Lambda_i$, the random variable $X^i$ can be written as
	    \begin{equation}\label{eq:Xi}
	      X^i= \sum_{k\in \bbN_i} \langle X,e_k^i\rangle e^i_k + X^{\bot,\Lambda_i}.
	    \end{equation}
	    Moreover, since $\Omega_i \in \G$, locality implies that $\rho(X){\bf 1}_{\Omega_i}= \rho(X{\bf 1}_{\Omega_i}){\bf 1}_{\Omega_i}$;
	    since $(e_k^i)_{k\in \bbN_i}$ is a sequence in $\Lambda_i$, we also have $e_k^i=e_k^i{\bf 1}_{\Omega_i}$ for every $k\in \bbN_i$. 
	    Hence, for fixed $k_0\in \bbN_i$, we have
	    \[
	       \langle \rho(X), e^i_{k_0} \rangle =
	        \langle \rho(X),{\bf 1}_{\Omega_i} e^i_{k_0} \rangle
	        =  \langle \rho(X){\bf 1}_{\Omega_i}, e^i_{k_0} \rangle 
	        = \langle \rho(X{\bf 1}_{\Omega_i}), e^i_{k_0} \rangle 
	        = \langle \rho(X^i), e^i_{k_0} \rangle.
	    \]
	    By \eqref{eq:Xi}, it follows that $\rho$ is local with respect to$(e_k^i)_{k\in \bbN_i,\,i\in \bbN}$. 
	\end{proof}
 
 We now examine a situation where the two notions of locality are equivalent.
 \begin{prop}\label{prop:equivalence}
    In addition to Assumption \ref{ass:structure}, suppose that 
$\mathcal{G}= \sigma(\Omega_i,i\in \bbN)$.
   Then $\rho$ is local with respect to $(e_k^i)_{k\in \bbN_i,\,i\in \bbN}$ if and only if it is local in the sense of Definition \ref{defn:properties risk measures}.
 \end{prop}
 
 \begin{proof}
 The backward implication is already given by Proposition \ref{prop:implication}. 
We prove the forward implication under the additional assumption that 
$\G=\sigma(\Omega_i:i\in \bbN)$. Suppose that $\rho$ is local with respect to $(e_k^i)_{k\in \bbN_i,\,i\in \bbN}$. Let $X\in L^2(\F)$ and $i\in\bbN$. To show that $\rho$ is local, it is enough to verify that $\rho(X){\bf 1}_{\Omega_i} = \rho(X {\bf 1}_{\Omega_i}){\bf 1}_{\Omega_i}$ thanks to the structure of $\G$. Note that $e_k^j{\bf 1}_{\Omega_i}=0$ for each $j\in\bbN\setminus\{i\}$, $k\in\bbN_j$, whereas $e_k^i{\bf 1}_{\Omega_i}=e_k^i$ for each $k\in\bbN_i$. Moreover, we may write
\[
\rho(X)= \sum_{j\in \bbN} \sum_{k\in \bbN_j} \langle \rho(X), e_k^j\rangle e_k^j,\quad \rho(X{\bf1}_{\Omega_i})= \sum_{j\in \bbN} \sum_{k\in \bbN_j} \langle \rho(X{\bf1}_{\Omega_i}), e_k^j\rangle e_k^j,
\]
Note that $X {\bf 1}_{\Omega_i}= \sum_{k\in \bbN_i} \langle X,e_k^i\rangle e^i_k + X^{\bot,\Lambda_i}$. Hence, by the locality of $\rho$ with respect to the basis,
\[
\rho(X){\bf 1}_{\Omega_i}=\sum_{k\in \bbN_i} \langle \rho(X), e_k^i\rangle e_k^i=\sum_{k\in \bbN_i} \langle \rho(X{\bf 1}_{\Omega_i}), e_k^i\rangle e_k^i = \rho(X{\bf 1}_{\Omega_i}){\bf 1}_{\Omega_i}.
\]
Therefore, $\rho$ is local in the sense of Definition \ref{defn:properties risk measures}.
 \end{proof}
	\begin{rem}
The assumptions of Proposition \ref{prop:equivalence} are satisfied by Examples \ref{ex:finite-dim} and \ref{ex:infinit-dim}.
\end{rem}
	
	Taking into account the results established in Propositions \ref{prop:implication} and \ref{prop:equivalence}, a natural question is if we can construct a risk measure $\rho$ that is local with respect to the given basis but not in the sense of Definition \ref{defn:properties risk measures}. The following example investigates this point.
	
	\begin{example}
	    Consider the setting of Example \ref{ex:finite-dim}, but with
	    \[
	    H_1:= \{ X\in L^2(\F): X{\bf 1}_{A_1} =0, X \text{ is } \sigma(A^1_1,A^2_1)\text{-measurable}\},
	    \]
	    where 
	    $A_1^1=\left\{\omega_1,\omega_2 \right\}$ and 
	    $A_1^2=\left\{\omega_3,\omega_4\right\}$. 
	    Assume further that $\G:= \sigma(A_1^1,A_1^2,A_2,A_3)$.
	    Then a basis for 
	    $H_1$ is given by
	    \[
	       e_1^1 = \frac{{\bf 1}_{A_1^1}}{\|{\bf 1}_{A_1^1}\|}, \quad   e_2^1 = \frac{{\bf 1}_{A_1^2}}{\|{\bf 1}_{A_1^2}\|}.
	    \]
	    Now set $\G^\prime:= \sigma(A_1,A_2,A_3)$ and $\rho(X)=\E[X|\G^\prime]$, $X\in  L^2(\F)$. Then, it turns out that 
	    \[
	      \rho(X){\bf 1}_{A_i} =  \rho(X{\bf 1}_{A_i}){\bf 1}_{A_i}, \ \ \ i\in\{1,2,3\},
	    \]
	    but 
	     \[
	      \rho(X){\bf 1}_{A_1^j} \neq \rho(X{\bf 1}_{A_1^j}){\bf 1}_{A_1^j}, \ \ \ j\in\{1,2\}.
	    \]
	    Hence, $\rho$ is local with respect to the basis $(e^1_1,e^1_2,e^2,e^3)$ but not local in the sense of Definition \ref{defn:properties risk measures}.
	\end{example}
	The main result of this subsection is Theorem \ref{thm:NQC equivalent to Conv locality wrt basis}, given below. It formulates the relationship between natural quasiconvexity and convexity for risk measures on $L^{2}(\F)$ that are local with respect to a decomposable basis. Before we get into the result, we need to introduce an assumption and some notations.
	
	  \begin{assumption}\label{as:3.4}
		In addition to Assumption \ref{ass:structure}, suppose that the following properties hold:
		\begin{enumerate}
		\item [a)] For each $i\in\bbN$, the subspace $H_i$ is $1$-dimensional, i.e., $|\bbN_i|=1$. In particular, $H_i= \spn\{ e^i\}$ for some random variable $e^i$. 
		\item [b)] For each $i\in\bbN$, the function $X\mapsto \langle \rho(X),e^i \rangle$ on $L^2(\F)$ is non-constant.
		\end{enumerate}
	\end{assumption}
	
		The following is an example of a risk measure for which Assumption \ref{as:3.4} holds.
	\begin{example}
		Under the setting of Assumption \ref{ass:structure}, assume that
		$\G = \sigma(\Omega_{i} \colon i\in \bbN)$, where $\P(\Omega_{i}) > 0$ for each $i\in\bbN$. Moreover, we set $H_i:=\{ X\in\Lambda_i: X \text{ is $\G$-measurable}\}$, which is a $1$-dimensional subspace.  Define $e^{i} = \frac{{\bf 1}_{\Omega_{i}}}{\sqrt{\P(\Omega_{i})}}$ for each $i\in \bbN$. Let us consider the certainty equivalent defined by
		\[
			\rho(X) = \ell^{-1}\E[\ell(-X)|\G],\quad X\in L^2(\F),
		\]
		where $\ell\colon \R\to\R$ is a continuous increasing loss function. Let $x_{1}, x_{2} \in \R$ such that $x_{1}\neq x_{2}$. We claim that Assumption \ref{as:3.4} holds for $\rho$ with $(e^{i})_{i\in\bbN}$ used as the basis. Observe that, for every $i\in \bbN$,
		\[
			\rho(\langle x_{1} , e^{i} \rangle e^{i} + x_{1}^{\perp}) = \rho(\E[x_{1}e^{i}]e^{i})
			= \ell^{-1}\E[\ell(-\E[x_{1}e^{i}]e^{i})|\G] 
			= \ell^{-1}\ell(-\E[x_{1}e^{i}]e_{i})
			 = -\E[x_{1}e^{i}]e^{i}.
		\]
		A similar calculation yields that $\rho(\langle x_{2} , e^{i} \rangle e^{i} + x_{2}^{\perp}) = -\E[x_{2}e^{i}]e^{i}$. Then, 
		\begin{align*}
			\langle \rho(\langle x_{1} , e^{i} \rangle e^{i} + x_{1}^{\perp}) ,e^{i} \rangle - \langle \rho(\langle x_{2} , e^{i} \rangle e^{i} + x_{2}^{\perp}) ,e^{i} \rangle
			&=\E[x_{2}e^{i}]e^{i} -\E[x_{1}e^{i}]e^{i} \\
			&= (x_{2} - x_{1})\E[e^{i} ] 			
			=(x_{2} - x_{1})\sqrt{\P(\Omega_{i})}
			> 0,
		\end{align*}
		where the strict inequality holds since $x_{1} \neq x_{2}$ and $\P(\Omega_{i}) >0$. Hence, $X  \mapsto \langle \rho(\langle X , e^{i} \rangle e^{i} + X^{\perp}) ,e^{i} \rangle$ is non-constant.
	\end{example}
	
	As a preparation for Theorem \ref{thm:NQC equivalent to Conv locality wrt basis}, for every $n\in\bbN$, let us define
	\begin{align*}
		\H_{n} &\coloneqq \spn\{e^1,\ldots,e^n\}=H_1\oplus\ldots\oplus H_n,
		 \\
		  \mathcal{K}_n &\coloneqq \cl\spn\{\beta^i_{j},\ j \in \bbM_i, \ i\in\{1,\ldots,n\}\}=H_1^{\bot,\Lambda_1}\oplus\ldots\oplus H_n^{\bot,\Lambda_n},
	\end{align*}
	where the closure is with respect to the norm topology of $L^2(\F)$. Observe that
	\[
		L^{2}(\F) = \cl\bigcup_{n= 1}^{\infty} \of{\H_{n} \oplus \mathcal{K}_n}.
	\]
	Next, we define a preorder $\preceq$ on $L^{2}(\G)$ by
	\[
		Y \preceq V \iff \forall i\in\bbN\colon \langle Y, e^i \rangle \leq  \langle V, e^i \rangle.
	\]
	The \emph{ordering cone} corresponding to $\preceq$ is given as
	\[
		C \coloneqq \{ Y \in L^{2}(\G) \colon 0\preceq Y \}.
	\]
	In other words, for every $Y,V\in L^2(\G)$, we have $Y\preceq V$ if and only if $V-Y\in C$. Furthermore, the (positive) dual cone of $C$ is defined as
	\[
		C^{+} \coloneqq \{ Y \in L^{2}(\G) \colon \langle Y, V \rangle \geq  0  \text{ for every }V\in C \}.
	\]
	Finally, we introduce finite-dimensional analogues of $\preceq$, $C$, $C^+$. For every $n\in \N$, let us define a preorder $\preceq_n$ on $\mathcal{H}_n $ by
	\[
	Y \preceq_n V \iff \forall i\in\{1,\ldots,n\}\cap\bbN\colon \langle Y, e_i \rangle \leq  \langle V, e_i \rangle.
	\]
Then, the corresponding ordering cone and its dual are given by
\[
   C_n =\{ Y \in \H_n \colon 0 \preceq_n Y  \},\qquad 
		C_n^{+} = \{ Y \in \H_n\colon \langle Y, V \rangle \geq  0  \text{ for every }V\in C_n \}.
\]
	
	\begin{rem}\label{rem:cone}
	i) We observe that the following identities hold:
	    $C^+=C$ and $C^+_n= C_n$. In fact, let $Y \in C$. Then, by the definition of $C$, $\langle Y,e^i \rangle\geq 0$ and $\langle V, e^i\rangle\geq 0$ for every $V\in C$. Hence, $\langle Y, V \rangle = \sum_{i\in\bbN} \langle Y,e^i\rangle  \langle V,e^i\rangle \geq 0$ for every $V\in C$, i.e., $Y\in C^+$. This shows $C\subseteq C^+$. Conversely, let $Y\in C^+$. For every $i\in \bbN$, since $e^i \in C$, we have $ \langle Y,e^i\rangle \geq 0$. Hence, $Y\in C$. This shows $C^+ \subseteq C$. Similarly, we can prove that $C_n=C^+_n$. \\
	 ii) The identities  $C^+=C$ and $C^+_n= C_n$ imply that
	    $$
	       Y \in C^+  \text{ (resp., $Y\in C^+_n$)} \quad \Leftrightarrow \quad \forall i\in\bbN\colon
	       \langle Y, e_i \rangle \geq  0  \text{ (resp., $\forall i\in\{1,\ldots,n\}\cap\bbN\colon  \langle Y, e_i \rangle \geq  0$)}.
	    $$
	\end{rem}

Using the preorder $\preceq$, we modify Definition \ref{StarQC defn} in the $L^2$ setting as follows: a functional $\rho\colon L^2(\F)\to L^2(\G)$ is called \emph{naturally quasiconvex with respect to $\preceq$} if for every $X,Y\in L^2(\F)$ and $\lambda\in [0,1]$ there exists $\mu\in[0,1]$ such that
\[
\rho(\lambda X+(1-\lambda) Y)\preceq \mu\rho(X)+(1-\mu)\rho(Y).
\]
Convexity with respect to $\preceq$ is defined similarly. Moreover, $\rho$ is called \emph{$\star$-quasiconvex with respect to $\preceq$} if $X\mapsto \ip{X,Y}$ is quasiconvex on $L^2(\F)$ for every $Y\in C^+$. Analogous to Theorem \ref{StarQC Equivalent to NQC}, it can be checked that a conditional risk measure $\rho$ is naturally quasiconvex with respect to $\preceq$ if and only if it is $\star$-quasiconvex with respect to $\preceq$.

We conclude this section with the main result.
	
	\begin{thm}\label{thm:NQC equivalent to Conv locality wrt basis}
		Under Assumption \ref{as:3.4}, suppose that $\rho$ is a conditional risk measure that is naturally quasiconvex with respect to $\preceq$, continuous with respect to the norm topologies on $L^2(\F)$ and $L^2(\G)$, normalized, and local with respect to $(e^i)_{i\in\bbN}$. Then, $\rho$ is convex with respect to $\preceq$.
	\end{thm}

	\begin{proof}
		Let $n \in \N$, and define $\rho_n$ to be the restriction of $\rho$ on $\H_n \oplus \mathcal{K}_n$.
		Observe that 
		\begin{equation}\label{finitesum}
		X\in \H_n \oplus \mathcal{K}_n
		  \quad \Rightarrow \quad \langle X, e^i \rangle =0,\  \langle X, \beta_j^i \rangle =0 \ \text{for all $i>n$, $j\in \bbM_i$}. 
		 \end{equation}
		 Let $X\in \H_n \oplus \mathcal{K}_n$. We have
		   \begin{align}\label{eq:rho-n}
		     \rho_n(X)= \rho(X) &= \sum_{i\in \bbN}\langle \rho(X), e^i\rangle e^i 
		     = \sum_{i\in \bbN}\Big\langle \rho \big(\langle X, e^i \rangle e^i + X^{\perp, \Lambda_i} \big), e^i\Big\rangle e^i \notag \\
		     & = \sum_{\substack{i\in\bbN\colon\\  i\leq n}}  \Big\langle \rho \big(\langle X, e^i \rangle e^i + X^{\perp, \Lambda_i} \big), e^i \Big\rangle e^i + \sum_{\substack{i\in\bbN\colon \\  i> n}}\langle \rho (0), e^i\rangle e^i \notag \\
		     &= \sum_{\substack{i\in\bbN\colon\\  i\leq n}}  \Big\langle \rho \big(\langle X, e^i \rangle e^i + X^{\perp, \Lambda_i} \big), e^i \Big\rangle e^i,
		  \end{align}
		  where the third equality follows from locality with respect to $(e^i)_{i\in \bbN}$, the fourth one follows by \eqref{finitesum}, and the last one follows from normalization. Let $Y\in C_n^+$. Then, using \eqref{eq:rho-n}, we obtain
		\[
			\E[\rho_n(X)Y] 
			=\sum_{\substack{i\in\bbN\colon\\ i\leq n}} \langle Y, e_{i} \rangle  \langle \rho (X), e_{i} \rangle 
			= \sum_{i=1}^n \langle Y, e_{i} \rangle   \Big\langle \rho \big(\langle X, e^i \rangle e^i + X^{\perp, \Lambda_i} \big), e^i \Big\rangle.
		\]
	    Hence, the function $X \mapsto \E[\rho_n(X)Y]$ is a finite decomposable sum on the space $\H_n \oplus \mathcal{K}_n$. This function is also quasiconvex since $\rho$ is $\star$-quasiconvex with respect to $\preceq$, as noted before the statement of the theorem.
		Moreover, by Remark \ref{rem:cone}ii), we have $\langle Y,e_i \rangle \geq 0$ holds for each $i\in \{1,\ldots,n\}$. Then, using Theorem \ref{MainResult-2.2}, we can argue as in the proof of Theorem \ref{Thm:NQC is equal to Convexity} and deduce that $\rho_n$ (hence $\rho$) is convex on $\H_n \oplus \mathcal{K}_n$. Since this is valid for each $n\in\N$, it follows that $\rho$ is convex on $\bigcup_{n\in\N} (\H_n \oplus \mathcal{K}_n)=\bigcup_{n\in\bbN} (\H_n \oplus \mathcal{K}_n)$.
		
		Finally, we prove that $\rho$ is convex on $L^2(\F)=\cl\bigcup_{n=1}^\infty (\H_n \oplus \mathcal{K}_n)$. Let $X,Z \in L^2(\F)$ and $\lambda \in [0,1]$. Then, there exist two sequences
		 $(X_k)_{k\in\N}$, $(Z_k)_{k\in\N}$ in $\bigcup_{n=1}^\infty (\H_n \oplus \mathcal{K}_n)$ such that $(X_k)_{k\in\N}$ converges to $X$ in $L^2(\F)$ and $(Z_k)_{k\in\N}$ converges to $Z$ in $L^2(\F)$. Then, for each $k\in\N$, we have
		 \[
		     \rho(\lambda X_k + (1-\lambda) Z_k) \leq \lambda \rho(X_k) + (1-\lambda) \rho(Z_k).
		 \]
		 Applying the continuity of $\rho$, we can easily pass to the limit as $k\to +\infty$ and obtain  
		  \[
		     \rho(\lambda X + (1-\lambda) Z) \leq \lambda \rho(X) + (1-\lambda) \rho(Z).
		  \]
		  Hence, $\rho$ is convex on $L^2(\F)$. 
	\end{proof}

	\section*{Acknowledgments}
	
	The first and second authors acknowledge the financial supports of the Department of Economics at Universit\`a degli Studi dell'Insubria and Istituto Nazionale di Alta Matematica ``Francesco Severi" during their visits to Varese. The second author acknowledges the financial support of Bilkent University during his graduate studies.

\singlespacing

\bibliographystyle{plainnat.bst}
\bibliography{ABMreflist.bib}

\end{document}